\documentclass{article}%
\pdfoutput=1
\usepackage{amsmath}
\usepackage{amsfonts}
\usepackage{amssymb}
\usepackage{graphicx}
\usepackage{mathtools}
\providecommand{\U}[1]{\protect\rule{.1in}{.1in}}
\newtheorem{theorem}{Theorem}

\newtheorem{conjecture}[theorem]{Conjecture}

\newtheorem{definition}[theorem]{Definition}

\newtheorem{lemma}[theorem]{Lemma}
\newtheorem{notation}[theorem]{Notation}

\newtheorem{proposition}[theorem]{Proposition}
\newtheorem{remark}[theorem]{Remark}

\newenvironment{proof}[1][Proof]{\noindent\textbf{#1.} }{\ \rule{0.5em}{0.5em}}
\begin{document}

\title{The large-$N$ limit for two-dimensional Yang--Mills theory}
\author{Brian C. Hall\thanks{Supported in part by NSF Grant DMS-1301534.}\\University of Notre Dame\\Department of Mathematics\\Notre Dame, IN 46556 USA\\bhall@nd.edu}
\date{}
\maketitle

\begin{abstract}
The analysis of the large-$N$ limit of $U(N)$ Yang--Mills theory on a surface
proceeds in two stages: the analysis of the Wilson loop functional for a
simple closed curve and the reduction of more general loops to a simple closed
curve. In the case of the 2-sphere, the first stage has been treated
rigorously in recent work of Dahlqvist and Norris, which shows that the
large-$N$ limit of the Wilson loop functional for a simple closed curve in
$S^{2}$ exists and that the associated variance goes to zero.

We give a rigorous treatment of the second stage of analysis in the case of
the 2-sphere. Dahlqvist and Norris independently performed such an analysis,
using a similar but not identical method. Specifically, we establish the
existence of the limit and the vanishing of the variance for arbitrary loops
with (a finite number of) simple crossings. The proof is based on the
Makeenko--Migdal equation for the Yang--Mills measure on surfaces, as
established rigorously by Driver, Gabriel, Hall, and Kemp, together with an
explicit procedure for reducing a general loop in $S^{2}$ to a simple closed
curve. The methods used here also give a new proof of these results in the
plane case, as a variant of the methods used by L\'{e}vy.

We also consider loops on an arbitrary surface $\Sigma$. We put forth two
natural conjectures about the behavior of Wilson loop functionals for
topologically trivial simple closed curves in $\Sigma.$ Under the weaker of
the conjectures, we establish the existence of the limit and the vanishing of
the variance for topologically trivial loops with simple crossings that
satisfy a \textquotedblleft smallness\textquotedblright\ assumption. Under the
stronger of the conjectures, we establish the same result without the
smallness assumption.

\end{abstract}
\setcounter{tocdepth}{2}

\tableofcontents

\setcounter{tocdepth}{2}

\section{Introduction and main results}

\subsection{The Makeenko--Migdal equation in two dimensions}

Let us fix a connected compact Lie group $K$ together with an Ad-invariant
inner product on its Lie algebra, $\mathfrak{k}.$ The path integral for
Euclidean Yang--Mills theory over a manifold $M$ is supposed to describe a
probability measure on the space of connections for a principal $K$-bundle
over $M.$ One of the main objects of study in such a theory is the
\textit{Wilson loop functional}, namely the expectation value of the trace (in
some fixed representation of $K$) of the holonomy of the connection around a
loop. The Makeenko--Migdal equation is an identity for the variation of Wilson
loop functionals with respect to a variation in the loop. The original version
of this equation, in any number of dimensions, was proposed by Makeenko and
Migdal in \cite{MM}. A version specific to the two-dimensional case was then
developed by Kazakov and Kostov in \cite[Eq. (24)]{KK}. (See also \cite[Eq.
(9)]{K} and \cite[Eq. (6.4)]{GG}.)

A special feature of the two-dimensional Yang--Mills measure is its invariance
under area-preserving diffeomorphisms. Suppose we fix the topological type of
a loop $L$ in a surface $\Sigma$ and consider the faces of $L,$ that is, the
connected components of the complement of $L$ in $\Sigma.$ Then the Wilson
loop functional depends only on the areas of the faces of $L.$ Let us now take
$K=U(N)$ with the inner product on the Lie algebra $u(N)$ given by the scaled
Hilbert--Schmidt inner product,
\begin{equation}
\left\langle X,Y\right\rangle \coloneqq N\mathrm{Trace}(X^{\ast}Y).
\label{innerProduct}%
\end{equation}
It is then convenient to express the Wilson loop functionals in terms of the
normalized trace,%
\begin{equation}
\mathrm{tr}(X)\coloneqq\frac{1}{N}\mathrm{Trace}(X). \label{normalizedTrace}%
\end{equation}
%

%TCIMACRO{\FRAME{ftbpFU}{1.9873in}{1.9839in}{0pt}{\Qcb{The labeling of the
%faces surrounding $v$}}{\Qlb{mmplot.fig}}{mmplot.eps}%
%width 1.9873in;  height 1.9839in;  depth 0pt;  original-width 3.2655in;
%original-height 3.2603in;  cropleft "0";  croptop "1";  cropright "1";
%cropbottom "0";  filename 'mmplot.eps';file-properties "XNPEU";}} }%
%BeginExpansion
\begin{figure}[ptb]%
\centering
\includegraphics[
height=1.9839in,
width=1.9873in
]%
{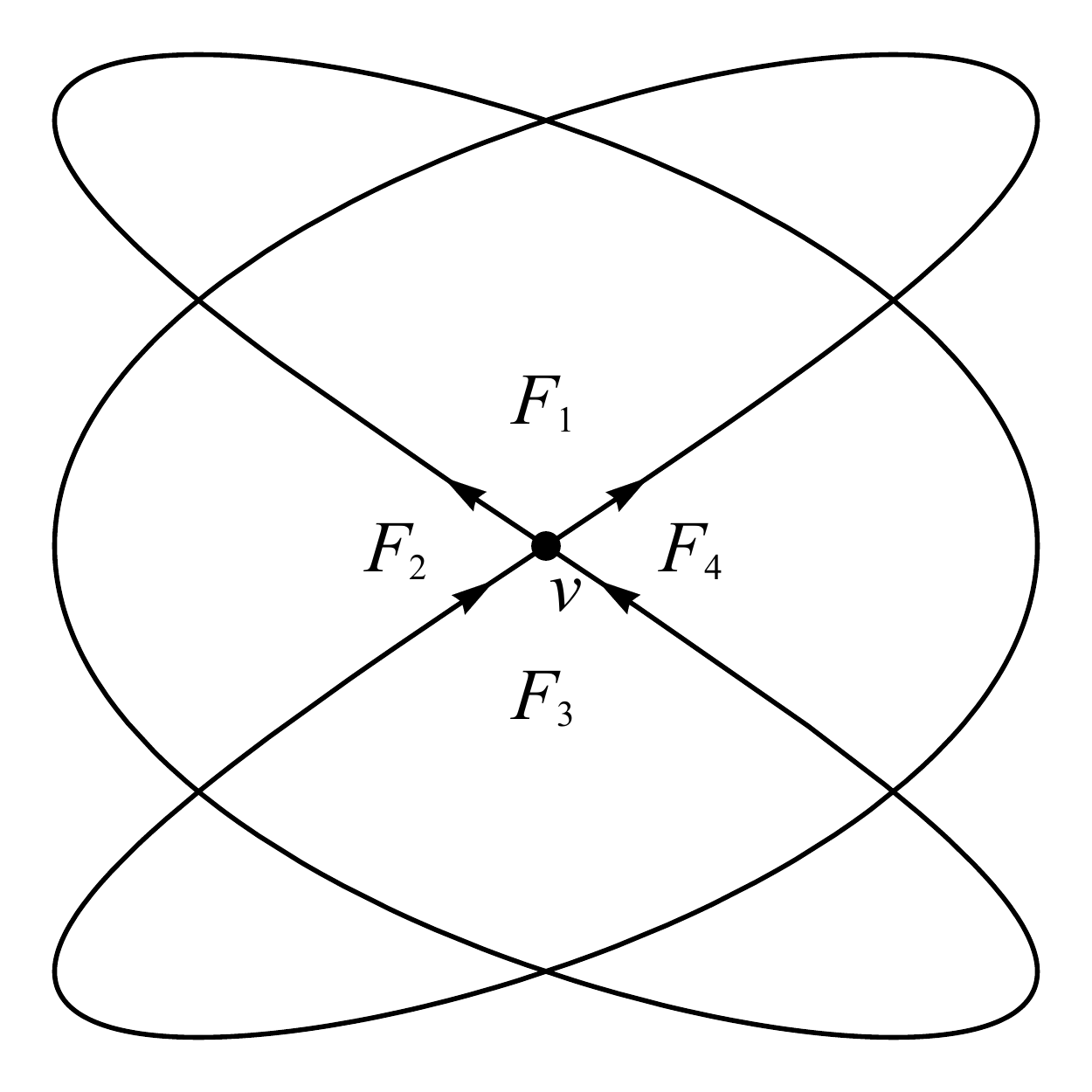}%
\caption{The labeling of the faces surrounding $v$}%
\label{mmplot.fig}%
\end{figure}
%EndExpansion

We now consider a loop $L$ with (a finite number of) simple crossings, and we
let $v$ be one such crossing. We label the four faces of $L$ adjacent to the
crossing in cyclic order as $F_{1},\ldots,F_{4},$ with $F_{1}$ denoting the
face whose boundary contains the two outgoing edges of $L.$ We then let
$t_{1},\ldots,t_{4}$ denote the areas of these faces. (See Figure
\ref{mmplot.fig}.) We also let $L_{1}$ denote the loop from the beginning to
the first return to $v$ and let $L_{2}$ denote the loop from the first return
to the end. (See Figure \ref{l1l2.fig}.)

The two-dimensional version of the Makeenko--Migdal equation, in the $U(N)$
case, is then as follows:%
\begin{equation}
\left(  \frac{\partial}{\partial t_{1}}-\frac{\partial}{\partial t_{2}}%
+\frac{\partial}{\partial t_{3}}-\frac{\partial}{\partial t_{4}}\right)
\mathbb{E}\{\mathrm{tr}(\mathrm{hol}(L))\}=\mathbb{E}\{\mathrm{tr}%
(\mathrm{hol}(L_{1}))\mathrm{tr}(\mathrm{hol}(L_{2}))\}, \label{MMUN}%
\end{equation}
where $\mathrm{hol}(\cdot)$ denotes the holonomy. Although the curves $L_{1}$
and $L_{2}$ occurring on the right-hand side of (\ref{MMUN}) are simpler than
the loop $L$, the right-hand side of (\ref{MMUN}) involves the
\textit{expectation of the product} of the traces, rather than the product of
the expectations. Thus, even if one has already computed the Wilson loop
functionals $\mathbb{E}\{\mathrm{tr}(\mathrm{hol}(L_{1}))\}$ and
$\mathbb{E}\{\mathrm{tr}(\mathrm{hol}(L_{2}))\},$ the right-hand side of
(\ref{MMUN}) cannot be regarded as a known quantity. In the large-$N$ limit,
however, we will see that the Makeenko--Migdal equation becomes an effective
tool for inductive computation of Wilson loop functionals.%

%TCIMACRO{\FRAME{ftbpFU}{2.1949in}{2.1949in}{0pt}{\Qcb{The loops $L_{1}$
%(black) and $L_{2}$ (dashed) for the loop in Figure \ref{mmplot.fig}}%
%}{\Qlb{l1l2.fig}}{l1l2.eps}{\special{ language "Scientific Word";
%type "GRAPHIC";  maintain-aspect-ratio TRUE;  display "USEDEF";
%valid_file "F";  width 2.1949in;  height 2.1949in;  depth 0pt;
%original-width 3.6115in;  original-height 3.6115in;  cropleft "0";
%filename 'l1l2.eps';file-properties "XNPEU";}} }%
%BeginExpansion
\begin{figure}[ptb]%
\centering
\includegraphics[
height=2.1949in,
width=2.1949in
]%
{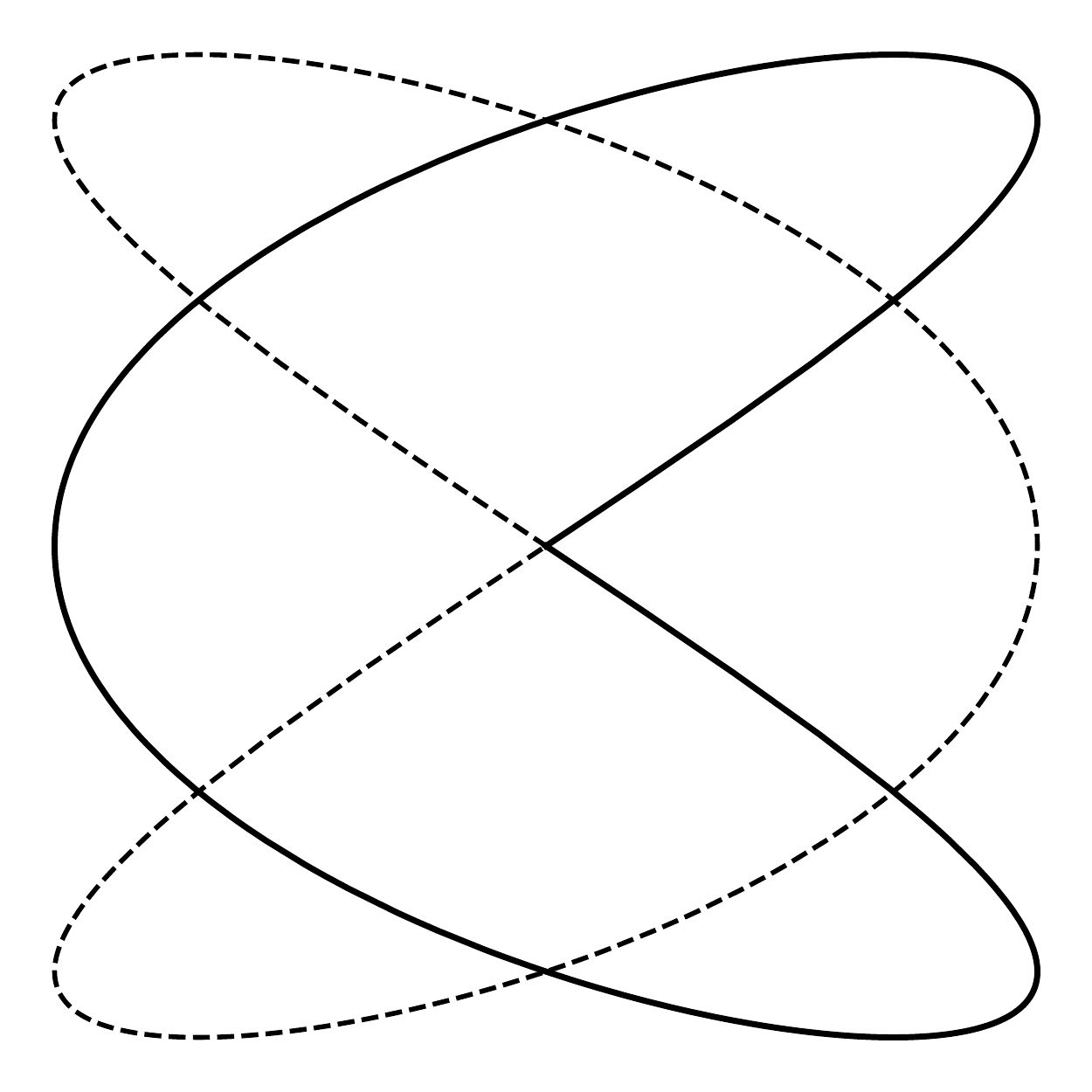}%
\caption{The loops $L_{1}$ (black) and $L_{2}$ (dashed) for the loop in Figure
\ref{mmplot.fig}}%
\label{l1l2.fig}%
\end{figure}
%EndExpansion

The original argument of Makeenko and Migdal for the equation that bears their
names was based on heuristic manipulations of the path integral. In the plane
case, L\'{e}vy then gave a rigorous proof of the Makeenko--Migdal equation in
\cite{LevyMaster}. (See Eq. (159) in Proposition 9.2.2 of \cite{LevyMaster}.)
Subsequent proofs of the planar Makeenko--Migdal equation were then provided
by Dahlqvist \cite{Dahl} and Driver--Hall--Kemp \cite{DHK2}.

Meanwhile, in \cite{DGHK}, Driver, Gabriel, Hall, and Kemp gave a rigorous
derivation of the Makeenko--Migdal equation for $U(N)$ Yang--Mills theory over
an arbitrary surface. Actually, the proof given in \cite{DHK2} in the plane
case extends with minor modifications to the case of a general surface.

\subsection{The master field in two dimensions}

In the paper \cite{tHooft}, 't Hooft proposed that Yang--Mills theory for
$U(N)$ in any dimension should simplify in the limit as $N\rightarrow\infty.$
In particular, it is expected that in this limit, the path integral should
concentrate onto a \textit{single} connection (modulo gauge transformations),
known as the master field. The concentration phenomenon for the Yang--Mills
measure has an important implication for the form of the two-dimensional
Makeenko--Migdal equation. Specifically, in the limit, there should be no
difference between the expectation of a product of traces and the product of
the associated expectations: both $\mathbb{E}\{fg\}$ and $\mathbb{E}%
\{f\}\mathbb{E}\{g\}$ should become $f(M_{0})g(M_{0}),$ where $M_{0}$ is the
master field.

If, therefore, the large-$N$ limit of $U(N)$ Yang--Mills theory exists on a
surface $\Sigma,$ we expect it to satisfy a Makeenko--Migdal equation of the
form%
\begin{equation}
\left(  \frac{\partial}{\partial t_{1}}-\frac{\partial}{\partial t_{2}}%
+\frac{\partial}{\partial t_{3}}-\frac{\partial}{\partial t_{4}}\right)
W(L)=W(L_{1})W(L_{2}), \label{mmLargeN}%
\end{equation}
where $W(L)$ is the limiting value of $\mathbb{E}\{\mathrm{tr}(\mathrm{hol}%
(L))\}.$ Note that the loops $L_{1}$ and $L_{2}$ on the right-hand side of
(\ref{mmLargeN}) have fewer crossings than $L,$ since neither $L_{1}$ nor
$L_{2}$ has a crossing at $v.$ Thus, one may hope that the large-$N$
Makeenko--Migdal equation may allow one to reduce computations of Wilson loop
functionals for general curves to simpler ones, until one eventually reaches a
simple closed curve. Of course, since a simple closed curve has no crossings,
the Makeenko--Migdal equation gives no information about the Wilson loop for
such a curve.

In the plane case, the structure of the master field was worked out by Singer
\cite{Si}, Gopakumar and Gross \cite{GG,Gop}, Xu \cite{Xu}, Sengupta
\cite{SenTraces}, Anshelevich and Sengupta \cite{AS}, and then in greater
detail by L\'{e}vy \cite{LevyMaster}. In particular, the expected
concentration phenomenon was verified in detail in the plane case in
\cite{LevyMaster}. (See the explicit variance estimate in Theorem 6.3.1 of
\cite{LevyMaster}.) A generalization of the master field on the plane was then
constructed by C\'{e}bron, Dahlqvist, and Gabriel in \cite{CDG}.

In \cite{LevyMaster}, L\'{e}vy shows that the large-$N$ limit of the Wilson
loop functional for a loop in the plane with simple crossings is completely
determined by (\ref{mmLargeN}), together with another, simpler condition. This
simpler condition---given as Axiom $\Phi_{4}$ on p. 11 of \cite{LevyMaster}
and called the \textquotedblleft unbounded face condition\textquotedblright%
\ in \cite[Theorem 2.3]{DHK2}---gives a simple formula for the derivative of
the Wilson loop functional with respect to the area of any face of $L$ that
adjoins the unbounded face.

\subsection{The master field on the sphere}

The existence of a large-$N$ limit of Yang--Mills theory on a general surface
$\Sigma$ is currently unknown. There has, however, been much interest in the
problem because of connections with string theory, as developed by Gross and
Taylor \cite{Gr,GT1,GT2}.

The $S^{2}$ case, meanwhile, has been extensively studied at varying levels of
rigor. The analysis proceeds in two stages. First, one studies the large-$N$
limit of the Wilson loop functional for a simple closed curve. Second, one
attempts to use the large-$N$ Makeenko--Migdal equation to reduce Wilson loop
functionals for all other loops with simple crossings to the simple closed curve.

In the first stage of analysis, a formula was proposed in the physics
literature for the Wilson loop functional for a simple closed curve. (See
Section \ref{wilsonSimple.sec} for more information.) A notable feature of
this formula is the presence of a phase transition. If the total area of the
sphere is less than $\pi^{2},$ the Wilson loop for a simple closed curve is
expressible in terms of the semicircular distribution from random matrix
theory. If, however, the total area is greater than $\pi^{2},$ the Wilson loop
is much more complicated. In addition to the proposed formula for the limiting
Wilson loop functional, it is expected that the limit should be deterministic,
in keeping with the idea of the master field. This brings us to the following
recent rigorous result of Dahlqvist and Norris \cite{DN}.

\begin{theorem}
[Dahlqvist--Norris]\label{dn.thm}If $C$ is a simple closed curve on $S^{2}$
then the limit%
\begin{equation}
\lim_{N\rightarrow\infty}\mathbb{E}\left\{  \mathrm{tr}(\mathrm{hol}%
(C))\right\}  \label{conj1}%
\end{equation}
exists and depends continuously on the areas of the two faces of $C.$
Furthermore, the associated variance tends to zero:%
\begin{equation}
\lim_{N\rightarrow\infty}\mathrm{Var}\left\{  \mathrm{tr}(\mathrm{hol}%
(C))\right\}  =0. \label{conj2}%
\end{equation}

\end{theorem}

The method of proof used in \cite{DN} is discussed briefly in Section
\ref{wilsonSimple.sec}.

\begin{notation}
\label{w1.notation}We denote the large-$N$ limit of the Wilson loop functional
for a simple closed curve by $W_{1}$:%
\[
W_{1}(a,b)=\lim_{N\rightarrow\infty}\mathbb{E}\left\{  \mathrm{tr}%
(\mathrm{hol}(C))\right\}  ,
\]
where $C$ is a simple closed curve and where $a$ and $b$ are the areas of the
faces of $C.$
\end{notation}

In the second stage of analysis, it has been claimed by Daul and Kazakov that,
\textquotedblleft All averages for self-intersecting loops can be reproduced
from the average for a simple (non-self-intersecting) loop by means of loop
equations.\textquotedblright\ (See the abstract of \cite{DaK}. The loop
equations referred to are the large-$N$ Makeenko--Migdal equation
(\ref{mmLargeN}).) It should be noted, however, that Daul and Kazakov analyze
only two examples, and it is not obvious how to extend their analysis to
general loops; see Section \ref{computing.sec}. Furthermore, they assume that
the large-$N$ limit exists and satisfies the large-$N$ Makeenko--Migdal equation.

\subsection{The reduction procedure}

In this paper, we give a rigorous treatment of the second stage of the
analysis of the large-$N$ limit for Yang--Mills theory on $S^{2},$ as well as
results for the plane and general surfaces. (Dahlqvist and Norris also give
treat the $S^{2}$ case by a similar but not identical method, as discussed
further in Section \ref{reductionSphere.sec}.)

\subsubsection{On the sphere\label{reductionSphere.sec}}

Specifically, we establish the following results in the sphere case: (1) the existence of the
large-$N$ limit of Wilson loop functionals for arbitrary loops with simple
crossings; (2) the vanishing of the associated variance; and (3) the large-$N$
Makeenko--Migdal equation for the limiting theory. In particular, we give a
concrete procedure for reducing the Wilson loop functional for general loops
in $S^{2}$ to the Wilson loop functional for a simple closed curve.

Here are some notable features of our approach.

\begin{itemize}
\item We do not assume the existence of the large-$N$ limit ahead of time,
except for a simple closed curve (Theorem \ref{dn.thm}).

\item We do not assume ahead of time that the limiting theory satisfies the
large-$N$ Makeenko--Migdal equation. Rather, we assume only the finite-$N$
Makeenko--Migdal equation in (\ref{MMUN}), as established rigorously in
\cite{DGHK}. We then prove that the limiting theory satisfies a large-$N$
version of the equation.

\item We give a constructive procedure for reducing the Wilson loop functional
an arbitrary loop in $S^{2}$ with simple crossings inductively to that for a
simple closed curve. Specifically, we show that any loop can first be reduced
to one that winds $n$ times around a simple closed curve, which can then be
reduced to a simple closed curve.
\end{itemize}

The just-referred-to procedure relies on a result (Proposition \ref{shrinkAllButTwo.prop}) that says that it is possible to perform a combination of Makeenko--Migdal variations at all of the vertices, with the effect that the areas of \textit{all but two} of the faces shrink to zero, with the areas of the remaining two faces remaining non-negative. Furthermore, it is possible to choose one of the ``unshrunk'' faces arbitrarily. 

After the first version of this paper was posted to the arXiv, I became aware
of a preprint of Dahlqvist and Norris \cite{DN}, which had been posted
approximately two months earlier. The paper of Dahlqvist and Norris proves
Theorem \ref{dn.thm}, which I stated as a conjecture in the first version of
this paper. In addition, \cite{DN} gives a reduction procedure that is similar
to, but not identical to, the one I use here. The main difference between the
two approaches is the just-mentioned freedom in my approach to arbitrarily choose one of the
faces whose area does not shrink to zero. This flexibility is exploited
crucially to give results on arbitrary surfaces, as discussed in Sections
\ref{surfaceIntro.sec} and \ref{surfaces.sec}.

Our main result on the sphere may be stated as follows.%

%TCIMACRO{\FRAME{ftbpFU}{1.9804in}{1.9804in}{0pt}{\Qcb{A checkerboard variation
%of the areas}}{\Qlb{lt.fig}}{lt.eps}{\special{ language "Scientific Word";
%type "GRAPHIC";  maintain-aspect-ratio TRUE;  display "USEDEF";
%valid_file "F";  width 1.9804in;  height 1.9804in;  depth 0pt;
%original-width 3.4861in;  original-height 3.4861in;  cropleft "0";
%filename 'lt.eps';file-properties "XNPEU";}} }%
%BeginExpansion
\begin{figure}[ptb]%
\centering
\includegraphics[
height=1.9804in,
width=1.9804in
]%
{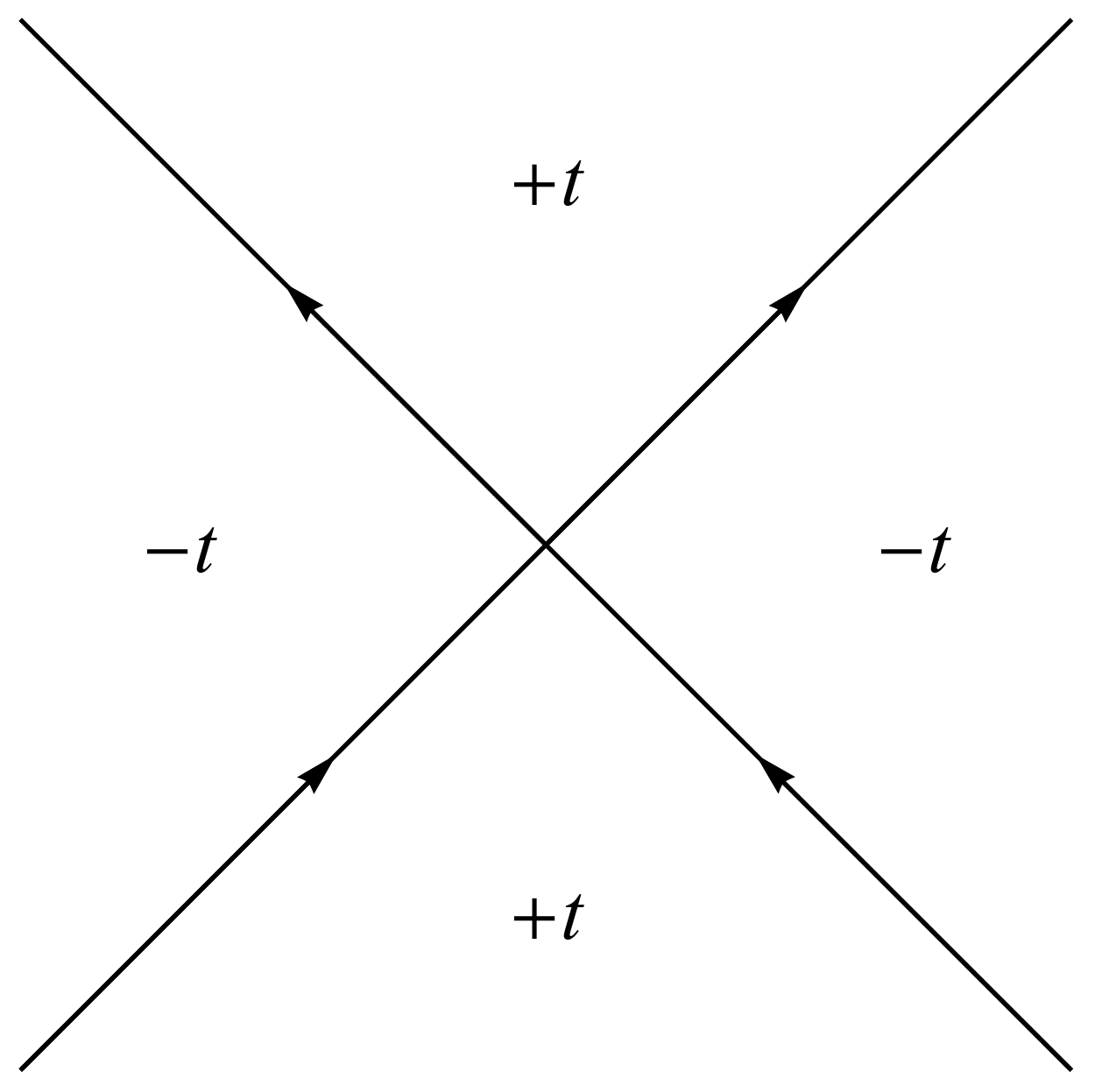}%
\caption{A checkerboard variation of the areas}%
\label{lt.fig}%
\end{figure}
%EndExpansion

\begin{theorem}
\label{main.thm}If $L$ is a closed curve traced out on a graph in $S^{2}$ and
having only simple crossings, the following results hold. First, the limit%
\begin{equation}
W(L):=\lim_{N\rightarrow\infty}\mathbb{E}\left\{  \mathrm{tr}(\mathrm{hol}%
(L))\right\}  \label{main1}%
\end{equation}
exists and depends continuously on the areas of the faces of $L.$ Second, the
associated variance goes to zero:
\begin{equation}
\lim_{N\rightarrow\infty}\mathrm{Var}\left\{  \mathrm{tr}(\mathrm{hol}%
(L))\right\}  =0. \label{main2}%
\end{equation}
Third, the limiting expectation values satisfy the following large-$N$
Makeenko--Migdal equation. Let us vary the areas of the faces surrounding a
crossing $v$ in a checkerboard pattern as in Figure \ref{lt.fig}, resulting in
a family of curves $L(t).$ Then%
\begin{equation}
\frac{d}{dt}W(L(t))=W(L_{1}(t))W(L_{2}(t)), \label{main3}%
\end{equation}
where $L_{1}(t)$ and $L_{2}(t)$ are derived from $L(t)$ in the usual way.
\end{theorem}

In Figure \ref{lt.fig}, we do not assume the four faces are distinct. If, say,
the two faces labeled as $+t$ are the same, we are then increasing the area of
that face by $2t.$

The reason for stating the Makeenko--Migdal equation in the form in
(\ref{main3}) is that we have not established the differentiability of the
large-$N$ Wilson loop functional $W(L)$ with respect to the area of an
individual face. If this differentiability property turns out to hold, we can
then apply the chain rule to express the derivative on the left-hand side of
(\ref{main3}) in the usual form as an alternating sum of such derivatives.
This issue is of little consequence, since the result in (\ref{main3}) is the
way one applies the Makeenko--Migdal equation in all applications.

\subsubsection{On the plane and on arbitrary surfaces\label{surfaceIntro.sec}}

We also provide a new proof of Theorem \ref{main.thm} in the plane case, as a
variant of the methods used by L\'{e}vy in \cite{LevyMaster}. In the plane
case, the result is not dependent on results of \cite{DN}, since the analog of
Theorem \ref{dn.thm} for $\mathbb{R}^{2}$ is a simple computation; see Section
\ref{planeGeneral.sec}. In L\'{e}vy's analysis in \cite{LevyMaster}, the
structure of the master field on $\mathbb{R}^{2}$ is based on two main axioms, the
large-$N$ version of the Makeenko--Migdal equation and a second condition,
labeled as Axiom $\Phi_{4}$ in \cite[Section 0]{LevyMaster} and called the
\textquotedblleft unbounded face condition\textquotedblright\ in \cite[Theorem
2.3]{DHK2}, which gives a formula for the derivative of a Wilson loop
functional with respect to the area of any face that adjoins the unbounded
face. (There are also some continuity and invariance properties.) We show that
the master field on $\mathbb{R}^{2}$ can alternatively be characterized by the
large-$N$ Makeenko--Migdal equation together with the (simple) formula for the
Wilson loop for a simple closed curve. See Section \ref{plane.sec}.

Finally, we consider Yang--Mills theory on an arbitrary compact surface
$\Sigma.$ Let us call a loop in $\Sigma$ \textquotedblleft topologically
trivial\textquotedblright\ if it is contained in a topological disk
$U\subset\Sigma.$ We put forth two natural conjectures regarding topologically
trivial simple closed curves in $\Sigma$. The first, Conjecture
\ref{surfaces.conjecture}, is simply the obvious analog of Theorem
\ref{dn.thm} for topologically trivial simple closed curves. (No such result
is known for surfaces other than the plane and the sphere.) The second,
Conjecture \ref{surfaces2.conjecture}, asserts also similar results for a loop
that winds $n$ times around a simple closed curve, $n\in\mathbb{Z}.$ Assuming
the first conjecture, we establish the analog of Theorem \ref{main.thm} for
topologically trivial loops with simple crossings that satisfy a
\textquotedblleft smallness\textquotedblright\ assumption. Assuming the second
conjecture, we establish the analog of Theorem \ref{main.thm} for \textit{all}
topologically trivial loops with simple crossings. See Section
\ref{surfaces.sec}.

Our results for the plane and for arbitrary surfaces depend crucially on an
extra level of flexibility in the reduction process that is not present in
\cite{DN}. This paper and \cite{DN} both use an approach in which the areas of
\textit{all but two} of the faces of the curve are shrunk to zero. In the
procedure in \cite[Section 4.5]{DN}, one generically has no choice regarding
which two faces remain unshrunk; they are specified by conditions on the
winding numbers. In our approach, by contrast, \textit{one} of the unshrunk
faces can be chosen arbitrarily. In the plane case, we choose one of the
unshrunk faces to be the unbounded face, while in the case of an arbitrary
surface, we choose one of the unshrunk faces to be the one containing the
complement of the topological disk $U.$

\subsection{The Wilson loop for simple closed curve in $S^{2}$%
\label{wilsonSimple.sec}}

In this section, we describe three approaches (at varying levels of rigor) to
analyzing the Wilson loop functional for a simple closed curve in the sphere.
If $C$ is a simple closed curve on $S^{2}$ and the areas of the two faces of
$C$ are $a$ and $b,$ Sengupta's formula \cite{Sen93} reads%
\begin{equation}
\mathbb{E}\left\{  \mathrm{tr}(\mathrm{hol}(C))\right\}  =\frac{1}{Z}%
\int_{U(N)}\mathrm{tr}(U)\rho_{a}(U)\rho_{b}(U)~dU, \label{wilsonSimple}%
\end{equation}
where $Z=\rho_{a+b}(\mathrm{id})$ is a normalization factor. Here $\rho_{a}$
is the heat kernel on $U(N),$ based at the identity and evaluated at
\textquotedblleft time\textquotedblright\ $a.$ The probability measure
\begin{equation}
\frac{1}{Z}\rho_{a}(U)\rho_{b}(U)~dU \label{simpleMeasure}%
\end{equation}
is precisely the distribution at time $a$ of a Brownian bridge on $U(N),$
starting at the origin and returning to the origin at time $a+b.$

In the first approach, one writes the heat kernels in (\ref{wilsonSimple}) as
sums over the characters of the irreducible representations of $U(N).$ In the
large-$N$ limit, one attempts to find the \textquotedblleft most probable
representation,\textquotedblright\ that is, the one whose character
contributes the most to the sum. The representations, meanwhile, are labeled
by certain diagrams; the objective is then to determine the limiting shape of
the diagram for the most probable representation. Using this method,
physicists have found different shapes in the small-area phase (namely
$a+b<\pi^{2}$) and the large-area phase (namely $a+b>\pi^{2}$). (See works by
Douglas and Kazakov \cite{DoK} and Boulatov \cite{Bou}.)

At a rigorous level, Boutet de Monvel and Shcherbina \cite{BS} and L\'{e}vy
and Ma\"{\i}da \cite{LM} have analyzed the partition function (i.e., the
normalization factor $Z=\rho_{a+b}(\mathrm{id})$) by this method and confirmed
the existence of a phase transition at $a+b=\pi^{2}.$ Then, recently,
Dahlqvist and Norris \cite[Section 3]{DN} have given a rigorous analysis of
the Wilson loop functional using a rigorous version of the arguments in
\cite{DoK} and \cite{Bou}, leading to Theorem \ref{dn.thm}.

In the second approach, one writes the heat kernels in (\ref{wilsonSimple}) as
a sum over all geodesics connecting the identity to $U,$ using a formula
developed by \`{E}skin \cite{Es} and rediscovered by Urakawa \cite{Ur}. (This
formula is a Poisson-summed version of the formula as a sum of characters.)
When the quantities $a$ and $b$ in (\ref{wilsonSimple}) are small, the
contribution of the shortest geodesic dominates. Recall that we are using the
scaled Hilbert--Schmidt inner product (\ref{innerProduct}) on the Lie algebra
$u(N).$ Since the Laplacian scales oppositely to the inner product, the
Laplacian on $U(N)$ is scaled by a factor of $1/N$ compared to the Laplacian
for the unscaled Hilbert--Schmidt inner product. Thus, at a heuristic level,
the large-$N$ limit ought to be pushing us toward the small-time regime for
the heat kernels $\rho_{a}$ and $\rho_{b}.$ It is therefore possible that in
the large-$N$ limit, one can simply \textquotedblleft neglect the winding
terms,\textquotedblright\ that is, include only the contribution from the
shortest geodesic.

The contribution of the shortest geodesic, meanwhile, is a Gaussian integral
of the sort that arises in the Gaussian unitary ensemble (GUE) in random
matrix theory. Thus, \textit{if} it is valid to keep only the contribution
from the shortest geodesic, the Wilson loop functional may be computed using
results from GUE theory. (See the work of Daul and Kazakov in \cite{DaK}.) On
the other hand, a consistency argument indicates that neglecting the winding
terms can only be valid in the small area phase. Little work has been done,
however, in estimating the size of the winding terms.

In the third approach, one may, as we have noted, recognize the probability
measure in (\ref{simpleMeasure}) as the distribution of a Brownian bridge on
$U(N).$ Forrester, Majumdar, and Schehr have then developed a method
\cite{FMS} to represent the partition function for Yang--Mills theory in terms
of a collection of $N$ nonintersecting Brownian bridges on the unit circle.
(That is, we consider $N$ Brownian motions in the unit circle, starting at 1.
We then constrain them to return to 1 at time $T=a+b$ and to be
nonintersecting for all times $0<t<T.$) In fact, the distribution of the
eigenvalues of the Brownian bridge in $U(N)$ is precisely the distribution of
these nonintersecting Brownian bridges. This claim is presumably well known to
experts---I learned it from Thierry L\'{e}vy---and is explained in the notes
\cite{Nonintersect}. (The claim is analogous to the well-known result that the
eigenvalues of a Brownian motion in the space of $N\times N$ Hermitian
matrices are described by the \textquotedblleft Dyson Brownian
motion\textquotedblright\ \cite{Dys} in $\mathbb{R}^{N}.$ See Section 3.1 of
\cite{Tao}.) Thus, not just the partition function, but also the Wilson loop
functional for a simple closed curve can be expressed in terms of
nonintersecting Brownian bridges.

Meanwhile, Liechty and Wang \cite{LW1,LW} have obtained various rigorous
results about the large-$N$ behavior of the nonintersecting Brownian bridges
in $S^{1}$. In particular, they confirm the existence of a phase transition:
When the lifetime $a+b$ of the bridge is less than $\pi^{2},$ the
nonintersecting Brownian motions do not wind around the circle, whereas for
lifetime greater than $\pi^{2}$ they do. It is possible that one could
establish Theorem \ref{dn.thm} rigorously in the small-area phase using
results from \cite{LW1}. (Theorem 1.2 of \cite{LW1} would be relevant.) In the
large-area phase, however, \cite{LW1} does not provide information about the
distribution of eigenvalues when $t$ is close to half the lifetime of the
bridge. (See the restrictions on $\theta$ in Theorem 1.5(a) of \cite{LW1}.)

\section{Tools for the proof}

In this section, we review some prior results that will allow us to prove our
main theorem. Our main tool, besides the crucial result of Dahlqvist and
Norris in Theorem \ref{dn.thm}, is the Makeenko--Migdal equation for $U(N)$
Yang--Mills theory on compact surfaces, which was established at a rigorous
level in \cite{DGHK}. More precisely, we require not only the standard
Makeenko--Migdal equation in (\ref{MMUN}), but also an \textquotedblleft
abstract\textquotedblright\ Makeenko--Migdal equation, which allows us to
compute the alternating sum of derivatives of expectation values of more
general functions. We also require an estimate on the variance of the product
of two bounded random variables, as described in Section \ref{varianceEst.sec}.

\subsection{Variation of the Wilson loop and of the
variance\label{variationVar.sec}}

Rigorous constructions of the two-dimensional Yang--Mills measure with
structure group $K$ from a continuum perspective were given in the plane case
by Gross, King, and Sengupta \cite{GKS} and by Driver \cite{Dr}, and in the
case of a compact surface, possibly with boundary, by Sengupta
\cite{Sen93,Sen97,Sen97b}. (See also \cite{LevSurfaces}, which, among other
things, extends the analysis to rectifiable loops with infinitely many
self-intersections.) In particular, suppose $\mathbb{G}$ is an
\textquotedblleft admissible\textquotedblright\ graph in a surface $\Sigma,$
meaning that $\mathbb{G}$ contains the boundary of $\Sigma$ and that each face
of $\mathbb{G}$ is a topological disk. Let $e$ denote the number of unoriented
edges of $\mathbb{G}$ and let $g$ be a gauge-invariant function of the
connection that depends only on the parallel transports $x_{1},\ldots,x_{e}$
along the edges of $\mathbb{G}.$ Then Driver (in the plane case) and Sengupta
(in the general case) give a formula for the expectation value of $g$ with
respect to the Yang--Mills measure. The formulas of Driver and Sengupta
correspond to what is known as the \textit{heat kernel action} in the physics
literature, as developed by Menotti and Onofri \cite{MO} and others.

Let $\rho_{t}:K\rightarrow\mathbb{R}$ denote the heat kernel on $K,$ based at
the identity. Then we have, explicitly,%
\begin{equation}
\mathbb{E}\left\{  g\right\}  =\frac{1}{Z}\int_{K^{e}}g(x_{1},\ldots
,x_{e})\prod_{i}^{{}}\rho_{\left\vert F_{i}\right\vert }(\mathrm{hol}%
(F_{i}))~dx_{1}\cdots dx_{e}, \label{senguptasFormula}%
\end{equation}
where $dx_{i}$ denotes the normalized Haar measure on $K$, $\left\vert
F_{i}\right\vert $ is the area of the $i$th face, and $\mathrm{hol}(F_{i})$ is
the product of edge variables going around the boundary of $F_{i}.$ Here $Z$
is a normalization constant. Since $\rho$ is invariant under conjugation and
inversion, the formula does not depend on the starting point or orientation of
the boundary of $F_{i}.$ If the boundary of $\Sigma$ is nonempty, it is
possible to incorporate into (\ref{senguptasFormula}) constraints on the
holonomies around the boundary components; the proof of the Makeenko--Migdal
equation in \cite{DGHK} holds in this more general context.

\begin{remark}
\label{varyAreas.remark}In the rest of the paper, when we speak about
\textquotedblleft varying the areas\textquotedblright\ of the faces of graph,
we mean more precisely that we replace the numbers $\{\left\vert F_{i}\right\vert
\}$ by some other collection of positive real numbers $\{t_{i}\}$ in
Sengupta's formula (\ref{senguptasFormula}). If the sum of the $t_{i}$'s
equals the sum of the $\left\vert F_{i}\right\vert $'s, it may be possible to
implement this variation \textquotedblleft geometrically,\textquotedblright%
\ by continuously deforming the graph, but this is not necessary. In
particular, the Makeenko--Migdal equation (\ref{MMUN}) was proved under such
an \textquotedblleft analytic\textquotedblright\ (i.e., not necessarily
geometric) variation of the area.

If we have a fixed loop $L$ and we let the areas of the faces of $L$ depend on
a parameter $t,$ we will (in a small abuse of notation) denote that pair
consisting of $L$ and the collection of areas by $L(t).$
\end{remark}

Suppose now that $L$ is a loop that can be traced out on an oriented graph in
$\Sigma$ and let $\mathbb{G}$ be a minimal graph on which $L$ can be traced.
We now explain what it means for $L$ to have a simple crossing at a vertex
$v$. First, we assume that $\mathbb{G}$ has exactly four edges incident to
$v$, where we count an edge $e$ twice if both the initial and final vertices
of $e$ are equal to $v$. Second, we assume that $L$, when viewed as a map of
the circle into the plane, passes through $v$ exactly twice. Third, we assume
that each time $L$ passes through $v$, it comes in along one edge and passes
\textquotedblleft straight across\textquotedblright\ to the cyclically
opposite edge. Last, we assume that $L$ traverses two of the edges on one pass
through $v$ and the remaining two edges on the other pass through $v$.

Under these assumptions, Theorem 1 of \cite{DGHK} gives a rigorous derivation
of the Makeenko--Migdal equation in (\ref{MMUN}). We now restate the
Makeenko--Migdal equation for $U(N),$ in the $S^{2}$ case, in a way that
facilitates the large-$N$ limit. In addition, we derive a similar result for
the variation of the variance of $\mathrm{tr}(\mathrm{hol}(L)),$ where for a
complex-valued random variable $X$, we define%
\[
\mathrm{Var}(X)=\mathbb{E}\{\left\vert X-\mathbb{E}\left\{  X\right\}
\right\vert ^{2}\}=\mathbb{E}\{\left\vert X\right\vert ^{2}\}-\left\vert
\mathbb{E}\left\{  X\right\}  \right\vert ^{2}.
\]

\begin{proposition}
\label{variation.prop}Let $L$ be a loop traced out on a graph in $S^{2}$ and
having only simple crossings. Let $v$ be one such crossing and let $L_{1}$ and
$L_{2}$ be obtained from $L$ as usual in the Makeenko--Migdal equation. Then
we have%
\begin{align}
&  \left(  \frac{\partial}{\partial t_{1}}-\frac{\partial}{\partial t_{2}%
}+\frac{\partial}{\partial t_{3}}-\frac{\partial}{\partial t_{4}}\right)
\mathbb{E}\{\mathrm{tr}(\mathrm{hol}(L))\}\nonumber\\
&  =\mathbb{E}\left\{  \mathrm{tr}(\mathrm{hol}(L_{1}))\right\}
\mathbb{E}\left\{  \mathrm{tr}(\mathrm{hol}(L_{2}))\right\} \nonumber\\
&  +\mathrm{Cov}\{\mathrm{tr}(\mathrm{hol}(L_{1})),\mathrm{tr}(\mathrm{hol}%
(L_{2}))\} \label{variationWilson}%
\end{align}
and
\begin{align}
&  \left(  \frac{\partial}{\partial t_{1}}-\frac{\partial}{\partial t_{2}%
}+\frac{\partial}{\partial t_{3}}-\frac{\partial}{\partial t_{4}}\right)
\mathrm{Var}\left\{  \mathrm{tr}(\mathrm{hol}(L))\right\} \nonumber\\
&  =2\operatorname{Re}\left[  \mathrm{Cov}\left\{  \mathrm{tr}(\mathrm{hol}%
(L_{1}))\mathrm{tr}(\mathrm{hol}(L_{2})),~\overline{\mathrm{tr}(\mathrm{hol}%
(L))}\right\}  \right] \nonumber\\
&  -\frac{1}{N^{2}}\mathbb{E}\left\{  \mathrm{tr}(\mathrm{hol}(L_{3}%
))\right\}  -\frac{1}{N^{2}}\mathbb{E}\left\{  \mathrm{tr}(\mathrm{hol}%
(L_{4}))\right\}  , \label{variationVariance}%
\end{align}
where $L_{3}$ and $L_{4}$ are the composite curves $L_{3}=L_{1}L_{2}L_{1}%
^{-1}L_{2}^{-1}$ and $L_{4}=L_{2}L_{1}L_{2}^{-1}L_{1}^{-1}.$ Here
$\mathrm{Cov}$ denotes the covariance, defined as $\mathrm{Cov}%
\{f,g\}=\mathbb{E}\left\{  fg\right\}  -\mathbb{E}\left\{  f\right\}
\mathbb{E}\left\{  g\right\}  .$
\end{proposition}

\begin{proof}
Equation (\ref{variationWilson}) is simply the Makeenko--Migdal equation
(\ref{MMUN}) rewritten using the definition of the covariance. To establish
(\ref{variationVariance}) we need to use the abstract Makeenko--Migdal
equation established in Theorem 2 of \cite{DGHK}. (This result generalizes the
abstract Makeenko--Migdal equation formulated and proved for the plane case by
L\'{e}vy in \cite[Proposition 9.1.3]{LevyMaster}.) Following the argument in
Section 2.3 of \cite{DHK2}, we express the loop $L$ as%
\[
L=e_{1}Ae_{4}^{-1}e_{2}Be_{3}^{-1},
\]
where $A$ and $B$ are words in edges other than $e_{1},\ldots e_{4}.$ Then
$L_{1}=e_{1}Ae_{4}^{-1}$ and $L_{2}=e_{2}Be_{3}^{-1}.$ If $a_{1},\ldots,a_{4}$
are the edge variables corresponding to $e_{1},\ldots,e_{4},$ we then have
(following the convention that parallel transport is order reversing)%
\[
\mathrm{hol}(L)=a_{3}^{-1}\beta a_{2}a_{4}^{-1}\alpha a_{1},
\]
where $\alpha$ and $\beta$ are words in the edge variables other than
$a_{1},\ldots,a_{4}.$

Now, the abstract Makeenko--Migdal equation reads%
\begin{equation}
\left(  \frac{\partial}{\partial t_{1}}-\frac{\partial}{\partial t_{2}}%
+\frac{\partial}{\partial t_{3}}-\frac{\partial}{\partial t_{4}}\right)
\mathbb{E}\left\{  f\right\}  =-\mathbb{E}\left\{  \nabla^{a_{1}}\cdot
\nabla^{a_{2}}f\right\}  , \label{abstractMM}%
\end{equation}
whenever $f$ has \textquotedblleft extended gauge invariance\textquotedblright%
\ at the vertex $v.$ When the edges $e_{1},\ldots,e_{4}$ are distinct,
extended gauge invariance at $v$ means that
\[
f(a_{1}x,a_{2},a_{3}x,a_{4},\mathbf{b})=f(a_{1},a_{2}x,a_{3},a_{4}%
x,\mathbf{b})=f(a_{1},a_{2},a_{3},a_{4},\mathbf{b})
\]
for all $x\in K,$ where $\mathbf{b}$ is the collection of edge variables other
than $a_{1},\ldots,a_{4}.$ (See Section 4 of \cite{DHK2} for a discussion of
extended gauge invariance when the edges are not distinct.)

Let us apply (\ref{abstractMM}) to the function $f=\left\vert g\right\vert
^{2},$ where%
\[
g(a_{1},a_{2},a_{3},a_{4},\mathbf{b})=\mathrm{tr}(\mathrm{hol}(L))=\mathrm{tr}%
(a_{3}^{-1}\beta a_{2}a_{4}^{-1}\alpha a_{1}).
\]
We also use the notation
\begin{align*}
g_{1}(a_{1},a_{2},a_{3},a_{4},\mathbf{b})  &  =\mathrm{tr}(\mathrm{hol}%
(L_{1}))=\mathrm{tr}(a_{4}^{-1}\alpha a_{1})\\
g_{2}(a_{1},a_{2},a_{3},a_{4},\mathbf{b})  &  =\mathrm{tr}(\mathrm{hol}%
(L_{2}))=\mathrm{tr}(a_{3}^{-1}\beta a_{2}).
\end{align*}
We then note that%
\begin{align*}
\nabla^{a_{1}}\cdot\nabla^{a_{2}}f  &  =\left(  \nabla^{a_{1}}\cdot
\nabla^{a_{2}}g\right)  \bar{g}+g(\nabla^{a_{1}}\cdot\nabla^{a_{2}}\bar{g})\\
&  +(\nabla^{a_{1}}g)\cdot(\nabla^{a_{2}}\bar{g})+(\nabla^{a_{2}}%
g)\cdot(\nabla^{a_{1}}\bar{g}).
\end{align*}

Now, as verified in \cite[Eqs. (2.13)--(2.15)]{DHK2}, we have $\nabla^{a_{1}%
}\cdot\nabla^{a_{2}}g=-g_{1}g_{2}$. Meanwhile, using $\overline{\mathrm{tr}%
(U)}=\mathrm{tr}(U^{-1})$ for $U\in U(N)$ and computing as in the second
example in \cite[Section 2.5]{DHK2}, we have that%
\begin{align*}
(\nabla^{a_{1}}g)\cdot(\nabla^{a_{2}}\bar{g})  &  =\frac{1}{N^{2}}%
\mathrm{tr}(\mathrm{hol}(L_{1}L_{2}L_{1}^{-1}L_{2}^{-1}))\\
(\nabla^{a_{2}}g)\cdot(\nabla^{a_{1}}\bar{g})  &  =\frac{1}{N^{2}}%
\mathrm{tr}(\mathrm{hol}(L_{2}L_{1}L_{2}^{-1}L_{1}^{-1})).
\end{align*}
Thus,%
\begin{align}
&  \left(  \frac{\partial}{\partial t_{1}}-\frac{\partial}{\partial t_{2}%
}+\frac{\partial}{\partial t_{3}}-\frac{\partial}{\partial t_{4}}\right)
\mathbb{E\{}\left\vert g\right\vert ^{2}\}\nonumber\\
&  =\mathbb{E}\left\{  g_{1}g_{2}\bar{g}\right\}  +\mathbb{E}\left\{
g\overline{g_{1}}\overline{g_{2}}\right\}  -\mathbb{E}\left\{  \mathrm{tr}%
(\mathrm{hol}(L_{3}))\right\}  -\mathbb{E}\left\{  \mathrm{tr}(\mathrm{hol}%
(L_{4}))\right\}  . \label{VarVar1}%
\end{align}
Meanwhile, by the ordinary Makeenko--Migdal equation (\ref{MMUN}), we have%
\begin{align}
&  \left(  \frac{\partial}{\partial t_{1}}-\frac{\partial}{\partial t_{2}%
}+\frac{\partial}{\partial t_{3}}-\frac{\partial}{\partial t_{4}}\right)
\left\vert \mathbb{E}\left\{  g\right\}  \right\vert ^{2}\nonumber\\
&  =\mathbb{E}\left\{  g_{1}g_{2}\right\}  \mathbb{E}\left\{  \bar{g}\right\}
+\mathbb{E}\left\{  g\right\}  \mathbb{E}\left\{  \overline{g_{1}}%
\overline{g_{2}}\right\}  . \label{VarVar2}%
\end{align}
Subtracting (\ref{VarVar2}) from (\ref{VarVar1}) gives the claimed result.
\end{proof}

\subsection{Variance estimates\label{varianceEst.sec}}

For a complex-valued random variable $X,$ we define the \textbf{variance} of
$X$ by%
\[
\mathrm{Var}(X):=\mathbb{E}\{\left\vert X-\mathbb{E}\left\{  X\right\}
\right\vert ^{2}\}=\mathbb{E}\{\left\vert X\right\vert ^{2}\}-\left\vert
\mathbb{E}\left\{  X\right\}  \right\vert ^{2}.
\]
In particular,
\begin{equation}
\mathrm{Var}(X)\leq\mathbb{E}\{\left\vert X\right\vert ^{2}\}.
\label{VarXLess}%
\end{equation}
We then define the \textbf{standard deviation} of $X,$ denoted $\sigma_{X}$,
by $\sigma_{X}=\sqrt{\mathrm{Var}(X)}.$

We observe that for any two random variables $X$ and $Y,$ we have%
\begin{equation}
\sigma_{X+Y}\leq\sigma_{X}+\sigma_{Y}, \label{sigmaIneq}%
\end{equation}
and similarly for any number of random variables. (It is harmless to assume
that the expectation values of $X$ and $Y$---and therefore $X+Y$---are zero,
in which case (\ref{sigmaIneq}) is the triangle inequality for the $L^{2}$
norm.) We also consider the covariance of two random variables, defined as%
\begin{align}
\mathrm{Cov}\{X,Y\}  &  =\mathbb{E}\left\{  (X-\mathbb{E}\left\{  X\right\}
)(Y-\mathbb{E}\left\{  Y\right\}  )\right\} \label{covariance0}\\
&  =\mathbb{E}\left\{  XY\right\}  -\mathbb{E}\left\{  X\right\}
\mathbb{E}\left\{  Y\right\}  . \label{covariance}%
\end{align}
(For our purposes, it is convenient not to put a complex conjugate into this
definition; the reader should note that $\mathrm{Cov}\{X,X\}\neq\mathrm{Var}(X)$.) We record the elementary inequality%
\begin{equation}
\left\vert \mathrm{Cov}\{X,Y\}\right\vert \leq\sigma_{X}\sigma_{Y},
\label{covIneq}%
\end{equation}
which follows from (\ref{covariance0}) and the Cauchy--Schwarz inequality.

We now establish a simple estimate on the standard deviation of the product of
two bounded random variables.

\begin{proposition}
\label{varianceProd.prop}Suppose $X$ and $Y$ are two complex-valued random
variables satisfying $\left\vert X\right\vert \leq1$ and $\left\vert
Y\right\vert \leq1$. Then%
\[
\sqrt{\sigma_{XY}}\leq\sqrt{\sigma_{X}}+\sqrt{\sigma_{Y}}.
\]

\end{proposition}

\begin{proof}
For simplicity of notation, let $\bar{X}=\mathbb{E}\left\{  X\right\}  $ and
let $\dot{X}=X-\bar{X}.$ Thus, $\mathrm{Var}(X)=\mathbb{E}\{|\dot{X}|^{2}\}.$
Then since $\left\vert X\right\vert \leq1,$ we have $\left\vert \bar
{X}\right\vert \leq1$ and $|\dot{X}|\leq2.$ Then by (\ref{VarXLess}) and the
Cauchy--Schwarz inequality, we have%
\begin{align*}
\mathrm{Var}(\dot{X}\dot{Y})  &  \leq\mathbb{E}\{|\dot{X}|^{2}|\dot{Y}%
|^{2}\}\\
&  \leq\sqrt{\mathbb{E}\{|\dot{X}|^{4}\}\mathbb{E}\{|\dot{Y}|^{4}\}}\\
&  \leq4\sigma_{X}\sigma_{Y},
\end{align*}
since $|\dot{X}|^{4}=|\dot{X}|^{2}|\dot{X}|^{2}\leq4|\dot{X}|^{2}$ and
similarly for $|\dot{Y}|^{4}.$ Now,%
\begin{align*}
\mathrm{Var}(XY)  &  =\mathrm{Var}(\dot{X}\dot{Y}+\bar{X}\dot{Y}+\dot{X}%
\bar{Y}+\bar{X}\bar{Y})\\
&  =\mathrm{Var}(\dot{X}\dot{Y}+\bar{X}\dot{Y}+\dot{X}\bar{Y})
\end{align*}
because adding a constant does not change the variance. Thus, by
(\ref{sigmaIneq}),
\begin{align*}
\sigma_{XY}  &  \leq\sigma_{\dot{X}\dot{Y}}+\left\vert \bar{X}\right\vert
\sigma_{\dot{Y}}+\left\vert \bar{Y}\right\vert \sigma_{\dot{X}}\\
&  \leq2\sqrt{\sigma_{X}\sigma_{Y}}+\sigma_{Y}+\sigma_{X}\\
&  =(\sqrt{\sigma_{X}}+\sqrt{\sigma_{Y}})^{2},
\end{align*}
which reduces to the claimed formula.
\end{proof}

\section{Examples\label{computing.sec}}

Before developing a general method for analyzing a general loop in $S^{2}$
with simple crossings, we consider two illustrative examples, the same two
that are considered in \cite{DaK}.

\subsection{The figure eight\label{figure8.sec}}

Although we consider loops on $S^{2},$ we can draw them as a loops on the
plane by picking a face and placing a puncture in that face, so that what is
left of $S^{2}$ is identifiable with $\mathbb{R}^{2}.$ We need to keep in
mind, however, that the \textquotedblleft unbounded\textquotedblright\ face in
such a drawing is actually a bounded face (with finite area) on $S^{2}.$
Furthermore, by placing the puncture in different faces, the same loop on
$S^{2}$ can give inequivalent loops on $\mathbb{R}^{2}.$ As a simple example,
consider the loop in Figure \ref{double.fig}, which we refer to as the figure
eight. Figure \ref{double.fig} gives two different views of this loop coming
from puncturing two different faces.

We write the holonomy around the figure eight as
\[
\mathrm{hol}_{L}(a,b,c)
\]
to indicate the dependence of the Wilson loop functional on the areas. In this
case, the loops $L_{1}$ and $L_{2}$ occurring on the right-hand side of the
Makeenko--Migdal equation are both simple closed curves. The loop $L_{1}$ (the outer loop on the left-hand side of Figure \ref{double.fig}, which is the lower loop on the right-hand side of the figure) encloses
areas $a+b$ and $c$. The loop $L_{2}$ (the inner loop on the left-hand side of the figure, which is the  upper loop on the right-hand side of the figure) encloses areas $b+c$ and $a$.%

%TCIMACRO{\FRAME{ftbpFU}{2.5278in}{1.3673in}{0pt}{\Qcb{Two views of a loop on
%$S^{2}$}}{\Qlb{double.fig}}{double.eps}{\special{ language "Scientific Word";
%type "GRAPHIC";  maintain-aspect-ratio TRUE;  display "USEDEF";
%valid_file "F";  width 2.5278in;  height 1.3673in;  depth 0pt;
%original-width 3.3053in;  original-height 1.7876in;  cropleft "0";
%filename 'double.eps';file-properties "XNPEU";}} }%
%BeginExpansion
\begin{figure}[ptb]%
\centering
\includegraphics[
height=1.3673in,
width=2.5278in
]%
{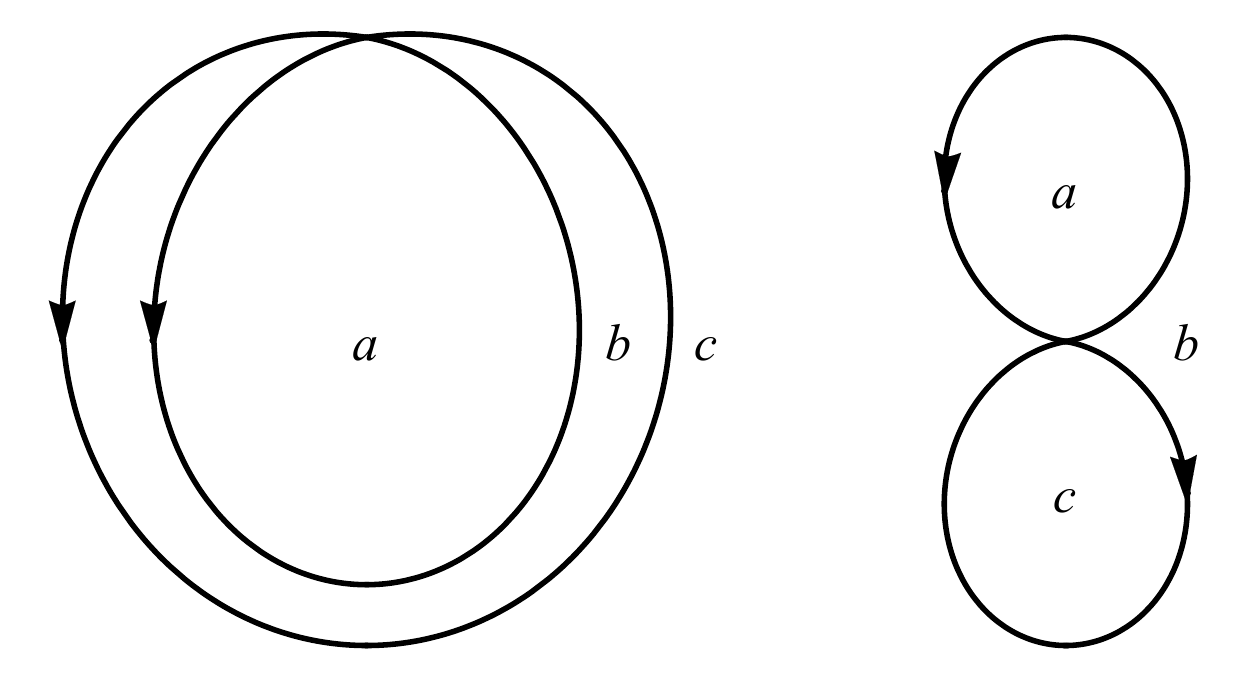}%
\caption{Two views of a loop on $S^{2}$}%
\label{double.fig}%
\end{figure}
%EndExpansion

\begin{theorem}
\label{figure8lim.thm}The limit%
\[
W_{L}(a,b,c):=\lim_{N\rightarrow\infty}\mathbb{E}\left\{  \mathrm{tr}%
(\mathrm{hol}_{L}(a,b,c))\right\}
\]
exists and satisfies the large-$N$ Makeenko--Migdal equation in the form%
\[
\left.  \frac{d}{dt}W_{L}(a-t,b+2t,c-t)\right\vert _{t=0}=W_{1}(a+b,c)W_{1}%
(a,b+c),
\]
where $W_{1}$ is as in Notation \ref{w1.notation}. Furthermore, we have
\[
\lim_{N\rightarrow\infty}\mathrm{Var}\left\{  \mathrm{tr}(\mathrm{hol}%
_{L}(a,b,c))\right\}  =0.
\]

\end{theorem}

See (\ref{figure8finite2}) for a formula for $W_L(a,b,c )$ in terms of the function $W_1(\cdot,\cdot)$.
If the partial derivatives of $W_{L}(a,b,c)$ with respect $a,$ $b,$ and $c$
exist and are continuous, it follows from the chain rule that
\[
\left.  \frac{d}{dt}W_{L}(a-t,b+2t,c-t)\right\vert _{t=0}=\left(
2\frac{\partial}{\partial b}-\frac{\partial}{\partial a}-\frac{\partial
}{\partial c}\right)  W_{L}(a,b,c).
\]
The following proof, however, does not establish the existence or continuity
of the partial derivatives of $W_{L}(a,b,c).$

\begin{proof}
We denote the holonomies for the two loops $L_{1}$ and $L_{2}$ as
$\mathrm{hol}_{L_{1}}(a+b,c)$ and $\mathrm{hol}_{L_{2}}(a,b+c)$. The face
labeled as $F_{1}$ should be the one bounded by the two outgoing edges of $L$
at $v,$ which is the face with area $b.$ Then $F_{3}$ coincides with $F_{1},$
while $F_{2}$ and $F_{4}$ are the faces with areas $a$ and $c$ (in either
order). We assume $c\geq a,$ with the case $c<a$ being entirely similar.

Proposition \ref{variation.prop} takes the form
\begin{align}
&  \frac{d}{dt}\mathbb{E}\left\{  \mathrm{tr}\left(  \mathrm{hol}%
_{L}(a-t,b+2t,c-t)\right)  \right\} \nonumber\\
&  =\left(  2\frac{\partial}{\partial b}-\frac{\partial}{\partial a}%
-\frac{\partial}{\partial c}\right)  \mathbb{E}\left\{  \mathrm{tr}\left(
\mathrm{hol}_{L}(a-t,b+2t,c-t)\right)  \right\} \nonumber\\
&  =\mathbb{E}\left\{  \mathrm{tr}\left(  \mathrm{hol}_{L_{1}}%
(a+b+t,c-t)\right)  \right\}  \mathbb{E}\left\{  \mathrm{tr}\left(
\mathrm{hol}_{L_{2}}(a-t,b+c+t)\right)  \right\} \nonumber\\
&  +\mathrm{Cov}\left\{  \mathrm{tr}\left(  \mathrm{hol}_{L_{1}}%
(a+b+t,c-t)\right)  ,\mathrm{tr}\left(  \mathrm{hol}_{L_{2}}%
(a-t,b+c+t)\right)  \right\}  . \label{figure8deriv2}%
\end{align}
Let us denote $\mathbb{E}\left\{  \mathrm{tr}\left(  \mathrm{hol}%
_{L}(a-t,b+2t,c-t)\right)  \right\}  $ by $F(t)$ (with $a,$ $b$, and $c$
fixed). We then write $F(0)=F(a-\varepsilon)-\int_{0}^{a-\varepsilon}%
F^{\prime}(t)~dt$; that is,%
\begin{align}
&  \mathbb{E}\left\{  \mathrm{tr}\left(  \mathrm{hol}_{L}(a,b,c)\right)
\right\} \nonumber\\
&  =\mathbb{E}\left\{  \mathrm{tr}\left(  \mathrm{hol}_{L}(\varepsilon
,2a+b-2\varepsilon,c-a+\varepsilon)\right)  \right\} \nonumber\\
&  -\int_{0}^{a-\varepsilon}\mathbb{E}\left\{  \mathrm{tr}\left(
\mathrm{hol}_{L_{1}}(a+b+t,c-t)\right)  \right\}  \mathbb{E}\left\{
\mathrm{tr}\left(  \mathrm{hol}_{L_{2}}(a-t,b+c+t)\right)  \right\}
~dt\nonumber\\
&  -\int_{0}^{a-\varepsilon}\mathrm{Cov}\left\{  \mathrm{tr}\left(
\mathrm{hol}_{L_{1}}(a+b+t,c-t)\right)  ,\mathrm{tr}\left(  \mathrm{hol}%
_{L_{2}}(a-t,b+c+t)\right)  \right\}  ~dt. \label{figure8finite}%
\end{align}

Now, it should be clear geometrically that if we let the area $a$ in the
figure eight tend to zero, the result is a simple closed curve. That is to
say, we expect that%
\begin{align}
&  \lim_{\varepsilon\rightarrow0}\mathbb{E}\left\{  \mathrm{tr}\left(
\mathrm{hol}_{L}(\varepsilon,2a+b-2\varepsilon,c-a+\varepsilon)\right)
\right\} \nonumber\\
&  =\mathbb{E}\left\{  \mathrm{tr}\left(  \mathrm{hol}_{L_{0}}%
(2a+b,c-a)\right)  \right\}  , \label{aLim}%
\end{align}
where $L_{0}(\alpha,\beta)$ is a simple closed curve on $S^{2}$ enclosing
areas $\alpha$ and $\beta.$ Analytically, (\ref{aLim}) follows easily from
Sengupta's formula, using that $\rho_{a}(\cdot)$ converges to a $\delta
$-measure on $K$ as $a$ tends to zero. Thus, letting $\varepsilon$ tend to
zero, we obtain%
\begin{align}
&  \mathbb{E}\left\{  \mathrm{tr}\left(  \mathrm{hol}_{L}(a,b,c)\right)
\right\} \nonumber\\
&  =\mathbb{E}\left\{  \mathrm{tr}\left(  \mathrm{hol}_{L_{0}}%
(2a+b,c-a)\right)  \right\} \nonumber\\
&  -\int_{0}^{a}\mathbb{E}\left\{  \mathrm{tr}\left(  \mathrm{hol}_{L_{1}%
}(a+b+t,c-t)\right)  \right\}  \mathbb{E}\left\{  \mathrm{tr}\left(
\mathrm{hol}_{L_{2}}(a-t,b+c+t)\right)  \right\}  ~dt\nonumber\\
&  -\int_{0}^{a}\mathrm{Cov}\left\{  \mathrm{tr}\left(  \mathrm{hol}_{L_{1}%
}(a+b+t,c-t)\right)  ,\mathrm{tr}\left(  \mathrm{hol}_{L_{2}}%
(a-t,b+c+t)\right)  \right\}  ~dt. \label{figure8finite2}%
\end{align}

Note that all holonomies on the right-hand side of (\ref{figure8finite2}) are
of simple closed curves. If we use (\ref{conj2}) in Theorem \ref{dn.thm}
together with the inequality (\ref{covIneq}) and dominated convergence, we
find that the last term in (\ref{figure8finite2}) tends to zero as $N$ tends
to infinity. Then using (\ref{conj1}) in Theorem \ref{dn.thm} along with
dominated convergence, we may let $N\rightarrow\infty$ to obtain%
\begin{align}
\lim_{N\rightarrow\infty}\mathbb{E}\left\{  \mathrm{tr}\left(  \mathrm{hol}%
_{L}(a,b,c)\right)  \right\}   &  =W_{1}(2a+b,c-a)\nonumber\\
&  -\int_{0}^{a}W_{1}(a+b+t,c-t)W_{1}(a-t,b+c+t)~dt, \label{figure8large}%
\end{align}
for $c\geq a,$ where $W_{1}$ is as in Notation \ref{w1.notation}. (Note that
the normalized trace defined in (\ref{normalizedTrace}) satisfies $\left\vert
\mathrm{tr}(U)\right\vert \leq1$ for all $U\in U(N),$ so that dominated
convergence applies in both integrals in (\ref{figure8finite2}).) This result
establishes the first claim in the theorem.

If we now subtract the value of (\ref{figure8large}) at $(a-s,b+2s,c-s)$ and
the value at $(a,b,c),$ we obtain%
\begin{align*}
&  W_{L}(a-s,b+2s,c-s)-W_{L}(a,b,c)\\
&  =\int_{0}^{s}W_{1}(a+b+t,c-t)W_{1}(a-t,b+c+t)~dt.
\end{align*}
Dividing this relation by $s$ and letting $s$ tend to zero gives%
\[
\left.  \frac{\partial}{\partial s}W_{L}(a-s,b+2s,c-s)\right\vert _{s=0}%
=W_{1}(a+b+t,c-t)W_{1}(a-t,b+c+t)
\]
by the continuity of $W_{1}.$ This relation is the desired large-$N$
Makeenko--Migdal equation.%

%TCIMACRO{\FRAME{ftbpFU}{3.5328in}{1.6276in}{0pt}{\Qcb{Trefoil with one lobe
%shrunk to area zero}}{\Qlb{trefoiloneshrunk1.fig}}{trefoiloneshrunk1.eps}%
%width 3.5328in;  height 1.6276in;  depth 0pt;  original-width 3.4861in;
%original-height 1.5912in;  cropleft "0";  croptop "1";  cropright "1";
%cropbottom "0";  filename 'trefoiloneshrunk1.eps';file-properties "XNPEU";}}
%}%
%BeginExpansion
\begin{figure}[ptb]%
\centering
\includegraphics[
height=1.6276in,
width=3.5328in
]%
{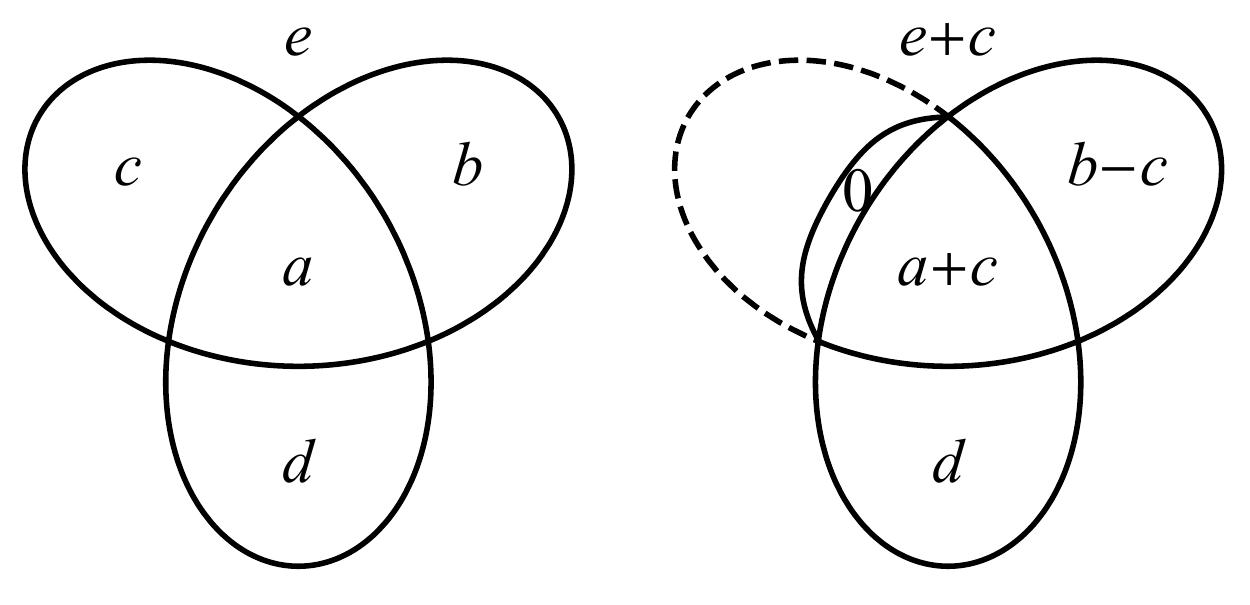}%
\caption{Trefoil with one lobe shrunk to area zero}%
\label{trefoiloneshrunk1.fig}%
\end{figure}
%EndExpansion

Meanwhile, by the second relation in Proposition \ref{variation.prop}, we have%
\begin{align}
&  \mathrm{Var}\left\{  \mathrm{tr}(\mathrm{hol}_{L}(a,b,c))\right\}
=\mathrm{Var}\left\{  \mathrm{tr}(\mathrm{hol}_{L_{0}}(b+2a,c-a))\right\}
\nonumber\\
&  -2\int_{0}^{a}\mathrm{Cov}\left\{  \mathrm{tr}(\mathrm{hol}(L_{1}%
))\mathrm{tr}(\mathrm{hol}(L_{2})),\mathrm{tr}(\mathrm{hol}(L))\right\}
~dt\nonumber\\
&  -\frac{2}{N^{2}}\int_{0}^{a}\mathbb{E}\left\{  \mathrm{tr}(\mathrm{hol}%
(L_{3}))\right\}  ~dt. \label{doubleVariance}%
\end{align}
The first term on the right-hand side tends to zero as $N$ tends to infinity,
by Theorem \ref{dn.thm}. In the second term on the right-hand side, $L_{1}$
and $L_{2}$ are simple closed curves. Furthermore, the normalized trace
satisfies $\left\vert \mathrm{tr}(U)\right\vert \leq1$ for all $U\in U(N).$
Thus, using (\ref{covIneq}) and Proposition \ref{varianceProd.prop} together
with Theorem \ref{dn.thm}, we see that the second term on the right-hand side
of (\ref{doubleVariance}) tends to zero. Finally, since $\left\vert
\mathrm{tr}(U)\right\vert \leq1,$ the last term on the right-hand side
manifestly goes to zero.
\end{proof}

%

%TCIMACRO{\FRAME{ftbpFU}{2.981in}{1.5471in}{0pt}{\Qcb{This loop is a
%perturbation of the loop on the right-hand side of Figure
%\ref{trefoiloneshrunk1.fig}}}{\Qlb{trefoiloneshrunk2.fig}}%
%{trefoiloneshrunk2.eps}{\special{ language "Scientific Word";
%type "GRAPHIC";  maintain-aspect-ratio TRUE;  display "USEDEF";
%valid_file "F";  width 2.981in;  height 1.5471in;  depth 0pt;
%original-width 3.474in;  original-height 3.3892in;  cropleft "0";
%filename 'trefoiloneshrunk2.eps';file-properties "XNPEU";}} }%
%BeginExpansion
\begin{figure}[ptb]%
\centering
\includegraphics[
height=1.5471in,
width=2.981in
]%
{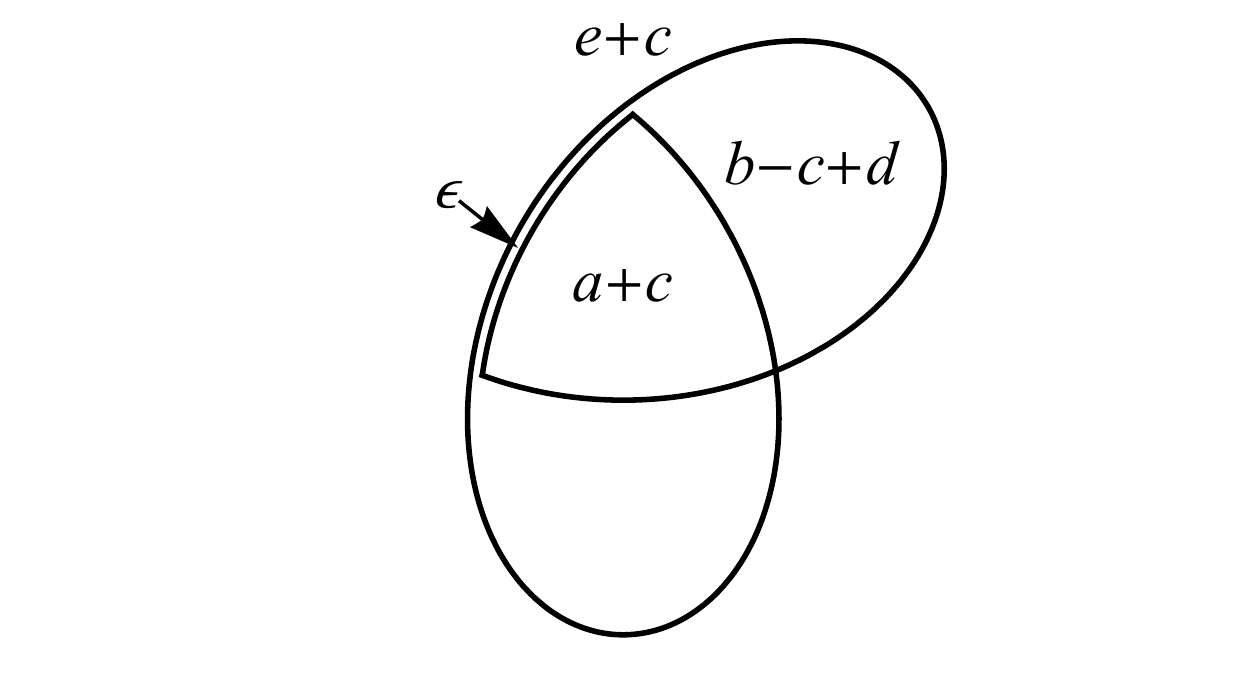}%
\caption{This loop is a perturbation of the loop on the right-hand side of
Figure \ref{trefoiloneshrunk1.fig}}%
\label{trefoiloneshrunk2.fig}%
\end{figure}
%EndExpansion

\subsection{The trefoil\label{trefoil.sec}}

In this section, we briefly outline an analysis of the trefoil loop by a
method similar to the one in the previous subsection. Later we will develop a
systematic method for analyzing any loop; this will provide an alternative
analysis of the trefoil. We label the areas of the faces as in Figure
\ref{trefoiloneshrunk1.fig}, where we may assume without loss of generality
that $c\leq b.$ We now perform a Makeenko--Migdal variation at the vertex in
the top middle of the figure. In this case, the loops $L_{1}$ and $L_{2}$ turn
out to be simple closed curves.

Let us denote by $L^{\prime}$ the loop on the right-hand side of Figure
\ref{trefoiloneshrunk1.fig}. Then $L^{\prime}$ is the limit as $\varepsilon$
tends to zero of the loops $L_{\varepsilon}^{\prime\prime}$ in Figure
\ref{trefoiloneshrunk2.fig}, where $\varepsilon$ denotes the distance between
the two nearby arcs of the loop. Specifically, in Figure
\ref{trefoiloneshrunk2.fig}, we let $\varepsilon$ tend to zero, while keeping
all of the areas of the faces fixed to the values indicated, using the
continuity properties of the Wilson loop functional developed in \cite[Theorem
2.58]{LevSurfaces}. But, $\mathbb{E}\left\{  \mathrm{tr}(\mathrm{hol}%
(L_{\varepsilon}^{\prime\prime}))\right\}  $ is independent of $\varepsilon$
and we conclude that $\mathbb{E}\left\{  \mathrm{tr}(\mathrm{hol}(L^{\prime
}))\right\}  =\mathbb{E}\left\{  \mathrm{tr}(\mathrm{hol}(L_{\varepsilon
}^{\prime\prime}))\right\}  .$ But since the loop $L_{\varepsilon}%
^{\prime\prime}$ is of the type analyzed in Section \ref{computing.sec}, we
may already know the large-$N$ behavior of $\mathbb{E}\left\{  \mathrm{tr}%
(\mathrm{hol}(L^{\prime}))\right\}  .$ The argument then proceeds much as in
Section \ref{figure8.sec}; since we will develop later a systematic method for
analyzing arbitrary loops, we omit the details of this analysis.

Other examples are not quite so easy to simplify by using the Makeenko--Migdal
equation at a single vertex. In the loop in Figure \ref{quad.fig}, for
example, it is not evident how shrinking any one of the faces to zero
simplifies the problem.

\section{Analysis of a general loop\label{general.sec}}

\subsection{The strategy\label{strategy.sec}}%

%TCIMACRO{\FRAME{ftbpFU}{2.3134in}{2.3091in}{0pt}{\Qcb{A more complicated
%loop}}{\Qlb{quad.fig}}{quad.eps}{\special{ language "Scientific Word";
%type "GRAPHIC";  maintain-aspect-ratio TRUE;  display "USEDEF";
%valid_file "F";  width 2.3134in;  height 2.3091in;  depth 0pt;
%original-width 3.2655in;  original-height 3.2603in;  cropleft "0";
%filename 'quad.eps';file-properties "XNPEU";}} }%
%BeginExpansion
\begin{figure}[ptb]%
\centering
\includegraphics[
height=2.3091in,
width=2.3134in
]%
{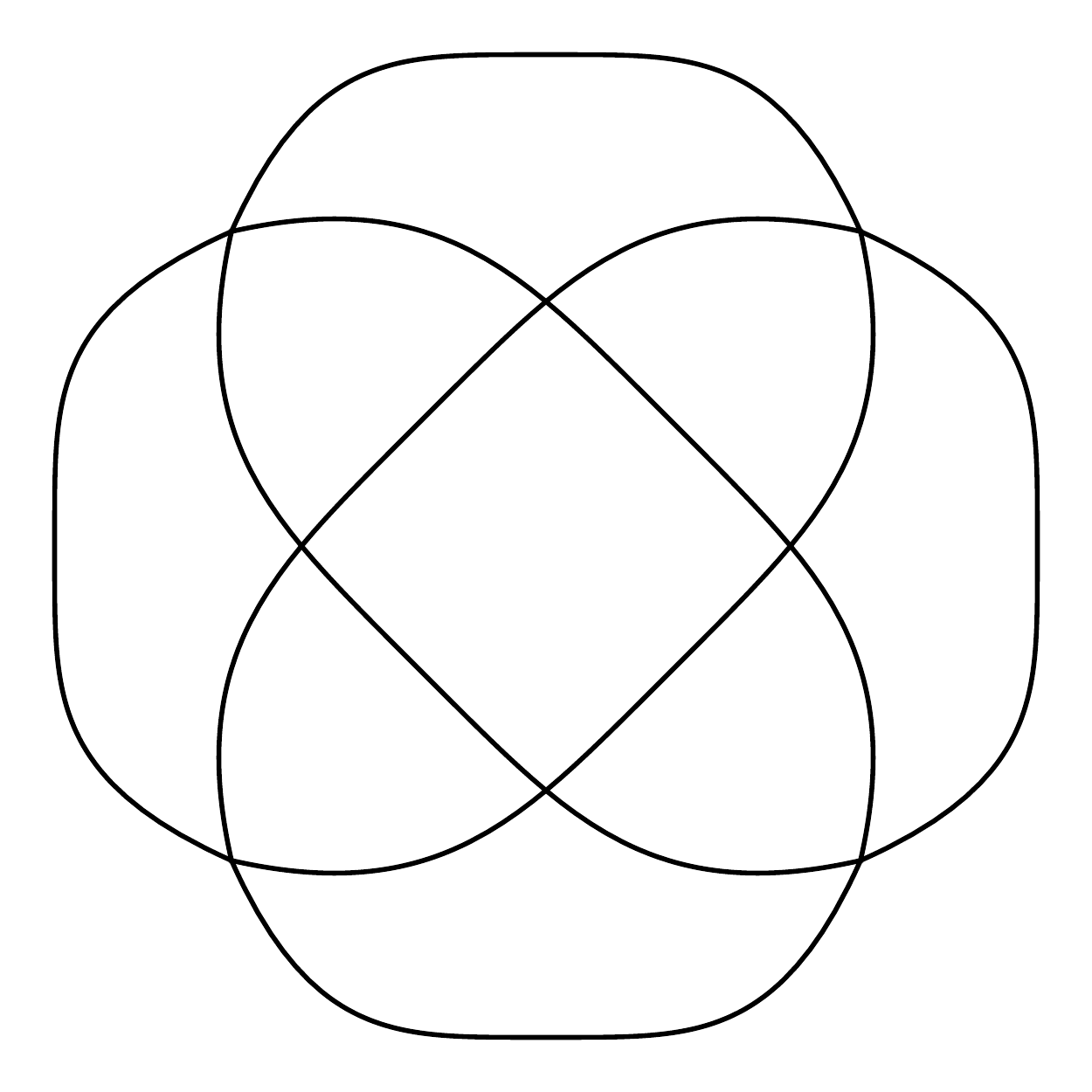}%
\caption{A more complicated loop}%
\label{quad.fig}%
\end{figure}
%EndExpansion

Given an arbitrary loop $L$ traced out on a graph in $S^{2}$ with simple
crossings, we will consider a linear combination of Makeenko--Migdal
variations of the areas over all the vertices of $L.$ We will show that it is
possible to make such a variation, depending on a parameter $t,$ in such a way
that as $t$ tends to 1, \textit{all but two} of the areas of the faces tend to
zero. Thus, in the $t\rightarrow1$ limit, we effectively have a loop with only
two faces. Indeed, we will show that the $t\rightarrow1$ limit of the Wilson
loop functional is the Wilson loop functional for a loop $L^{n}$ that winds
$n$ times around a simple closed curve. Here $n$ is an integer determined by
the topology of the original loop $L.$ (A similar but not identical procedure
is used in \cite{DN}. See Section \ref{selecting.sec} for a comparison.)%

%TCIMACRO{\FRAME{ftbpFU}{2.4777in}{2.4102in}{0pt}{\Qcb{We can shrink the areas
%of the three lobes of the trefoil to zero, while increasing the areas of the
%other two faces}}{\Qlb{trefoilshrink.fig}}{trefoilshrink.eps}%
%width 2.4777in;  height 2.4102in;  depth 0pt;  original-width 3.2655in;
%original-height 3.1782in;  cropleft "0";  croptop "1";  cropright "1";
%cropbottom "0";  filename 'trefoilshrink.eps';file-properties "XNPEU";}} }%
%BeginExpansion
\begin{figure}[ptb]%
\centering
\includegraphics[
height=2.4102in,
width=2.4777in
]%
{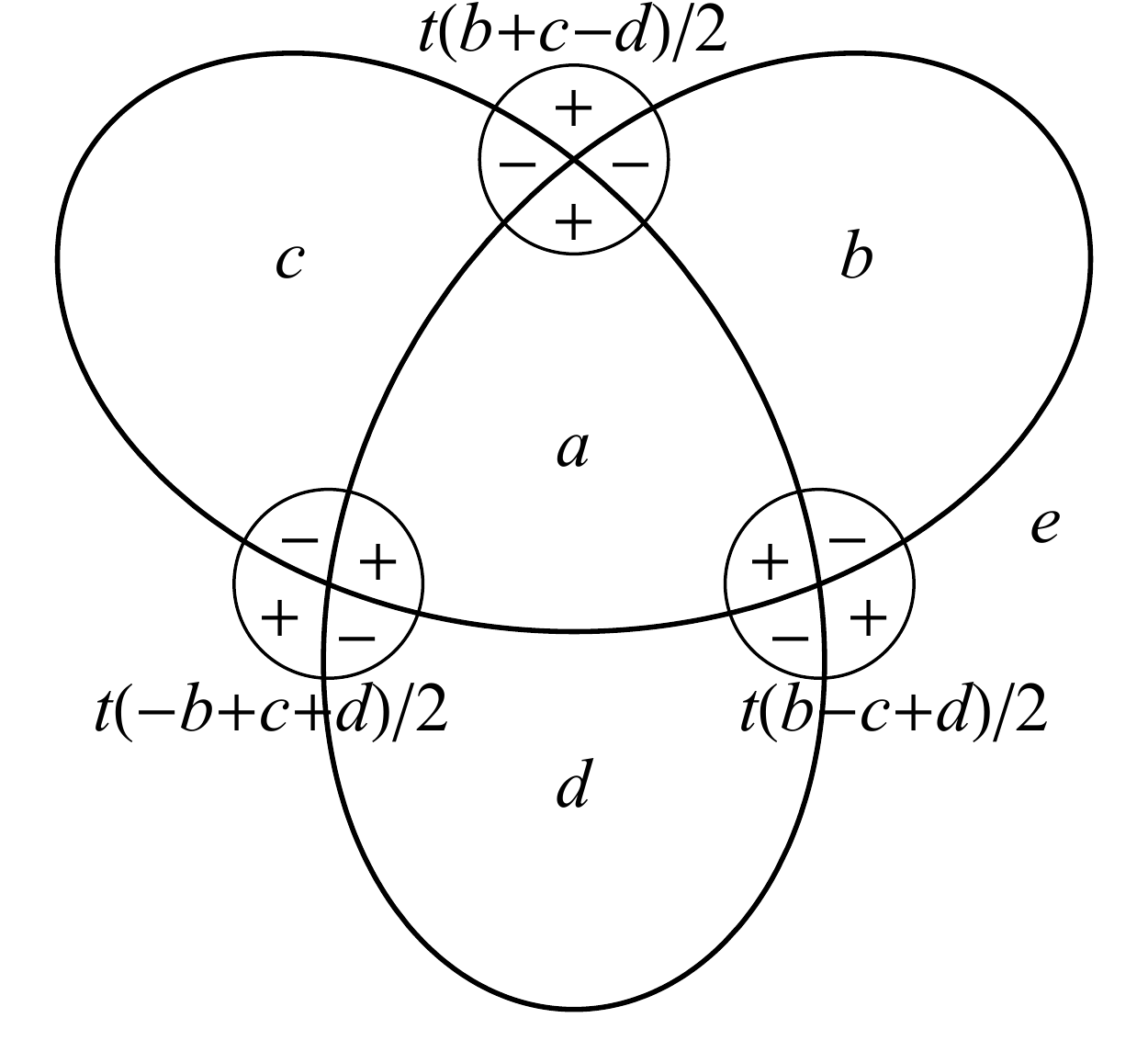}%
\caption{We can shrink the areas of the three lobes of the trefoil to zero,
while increasing the areas of the other two faces}%
\label{trefoilshrink.fig}%
\end{figure}
%EndExpansion

Consider, for example, the trefoil loop of Section \ref{trefoil.sec}. If we
vary the areas by the amounts shown in Figure \ref{trefoilshrink.fig}, the net
effect on the areas is:%
\[%
\begin{array}
[c]{cc}%
b\mapsto b-tb & \quad a\mapsto a+t(b+c+d)/2\\
c\mapsto c-tc & \quad e\mapsto e+t(b+c+d)/2\\
d\mapsto d-td &
\end{array}
.
\]
Thus, as $t$ varies from $0$ to $1,$ the areas $b,$ $c,$ and $d$ shrink
simultaneously to zero, while the areas of the two remaining faces increase.
The limiting curve is shown in Figure \ref{trefoilreduced.fig}. If the areas
of the three lobes are zero, the curve becomes a circle traversed twice.

Let $L(t)$ denote the trefoil with areas varying as above. If we differentiate
$\mathbb{E}\left\{  \mathrm{tr}(\mathrm{hol}(L(t)))\right\}  $ with respect to
$t,$ then by the chain rule, we will get a linear combination of terms of the
form
\[
\mathbb{E}\left\{  \mathrm{tr}(\mathrm{hol}(L_{1,j}(t)))\right\}
\mathbb{E}\left\{  \mathrm{tr}(\mathrm{hol}(L_{2,j}(t)))\right\}  ,
\]
where $L_{1,j}(t)$ and $L_{2,j}(t)$ represent the loop $L(t)$ cut at the $j$th
vertex of the trefoil ($j=1,2,3$), along with some covariance terms. Each
$L_{1,j}(t)$ and $L_{2,j}(t)$ is actually a simple closed curve, but in any
case, these curves have fewer crossings than the original trefoil. It is then
a straightforward matter to let $N$ tend to infinity to get the limiting
Wilson loop functional, and similarly for the variance.

For an arbitrary loop $L$ with $k$ crossings, we will show that we can deform
$L$ into a loop $L_{n}$ that winds $n$ times around a simple closed curve. The
variation of the Wilson loop functional along this path will be a linear
combination of products of Wilson loop functionals for curves with at most
$k-1$ crossings. In an inductive argument then, it remains only to analyze
$L_{n},$ which we do by another induction, this time reducing $L_{n}$ to
$L_{n-1},$ and so on.%

%TCIMACRO{\FRAME{ftbpFU}{2.4976in}{2.0384in}{0pt}{\Qcb{The trefoil with the
%lobes shrunk to area zero}}{\Qlb{trefoilreduced.fig}}{trefoilreduced.eps}%
%width 2.4976in;  height 2.0384in;  depth 0pt;  original-width 3.5284in;
%original-height 2.872in;  cropleft "0";  croptop "1";  cropright "1";
%cropbottom "0";  filename 'trefoilreduced.eps';file-properties "XNPEU";}} }%
%BeginExpansion
\begin{figure}[ptb]%
\centering
\includegraphics[
height=2.0384in,
width=2.4976in
]%
{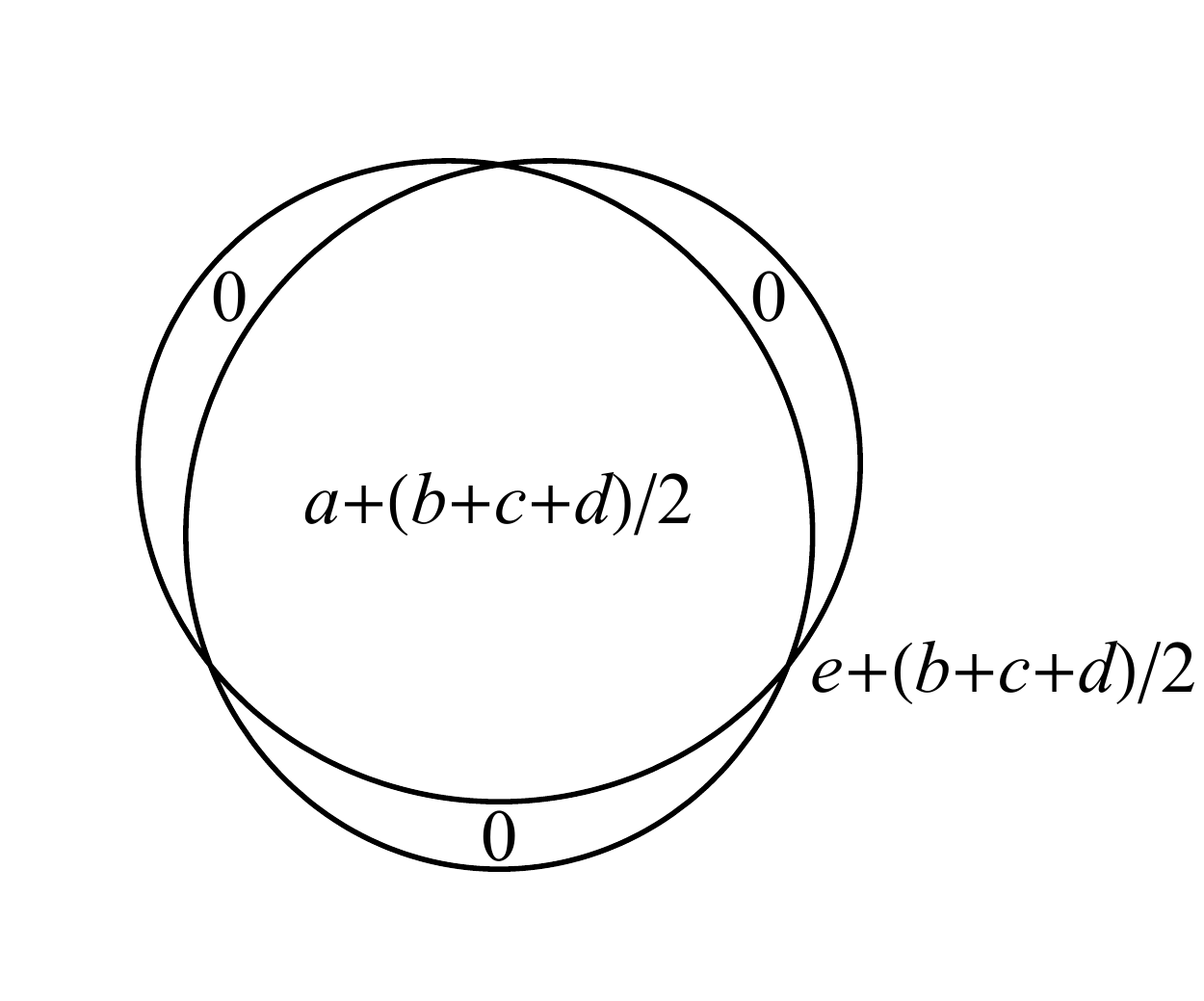}%
\caption{The trefoil with the lobes shrunk to area zero}%
\label{trefoilreduced.fig}%
\end{figure}
%EndExpansion

\subsection{Winding numbers}

We consider $S^{2}$ with a fixed orientation. We then consider a loop $L$
traced out on a graph in $S^{2}$ and having only simple crossings. We consider
the faces of $L,$ that is, the connected components of the complement of $L$
in $S^{2}.$ If we pick a face $F_{0},$ we can puncture $F_{0}$, thus turning
$S^{2}$ topologically into $\mathbb{R}^{2}.$ The orientation on $S^{2}$ gives
an orientation on $\mathbb{R}^{2}.$ Thus, for each face $F,$ we may speak
about the winding number of $L$ around $F.$ Since this winding number depends
on which face $F_{0}$ we puncture, we denote it thus:%
\[
w_{F_{0}}(F),
\]
so that $w_{F_{0}}(F_{0})=0.$ It is important to understand how the winding
number changes if the location of the puncture changes.

\begin{proposition}
\label{windingIndep.prop}If $F_{0},$ $F_{0}^{\prime},$ and $F$ are faces, then%
\[
w_{F_{0}}(F)-w_{F_{0}^{\prime}}(F)=w_{F_{0}}(F_{0}^{\prime}).
\]
In particular, the difference between $w_{F_{0}}(F)$ and $w_{F_{0}^{\prime}%
}(F)$ does not depend on $F.$
\end{proposition}

In particular, if we change the location of the puncture, all the winding
numbers change \textit{by the same amount}.

\begin{proof}
Let us fix points $x,$ $y,$ and $z$ in $F_{0},$ $F_{0}^{\prime},$ and $F,$
respectively. Let us put our puncture initially in $x,$ regarding
$S^{2}\setminus\{x\}$ as the plane. Let us then assume that $y$ is at the
origin and $z$ is at the point $(2,0)$. We may then regard $L$ as an element
of $\pi_{1}(\mathbb{R}^{2}\setminus\{y,z\})$ with base point at $(1,0),$ which
is a free group on two generators $e_{1}$ and $e_{2}$. These generators may be
identified with circles of radius one centered at $y$ and $z,$ respectively,
traversed in the counter-clockwise direction. Then $w_{F_{0}}(F)$ is the
number of times $L$ winds around $z,$ which is the number of occurrences of
the generator $e_{2}$ in the representation of $L$ as a word in $e_{1}$ and
$e_{2}.$

Suppose we now shift our puncture from $x$ to $y.$ This shift amounts to
composing $L$ with the inversion map in the complex plane, $\mathbb{C}%
=\mathbb{R}^{2},$ that is, the map $z\mapsto1/z.$ After this process, the
generator $e_{1}$ traverses the unit circle in the opposite direction, while
the generator $e_{2}$ is now inside the unit circle (Figure \ref{winding.fig}%
). Thus, $w_{F_{0}^{\prime}}(F)$ is the number of occurrences of $e_{2}$ minus
the number of occurrences of $e_{1},$ giving the claimed result.
\end{proof}

%

%TCIMACRO{\FRAME{ftbpFU}{3.8623in}{1.1658in}{0pt}{\Qcb{Transformation of the
%winding numbers under a change of puncture}}{\Qlb{winding.fig}}{winding.eps}%
%width 3.8623in;  height 1.1658in;  depth 0pt;  original-width 3.474in;
%original-height 1.0343in;  cropleft "0";  croptop "1";  cropright "1";
%cropbottom "0";  filename 'winding.eps';file-properties "XNPEU";}} }%
%BeginExpansion
\begin{figure}[ptb]%
\centering
\includegraphics[
height=1.1658in,
width=3.8623in
]%
{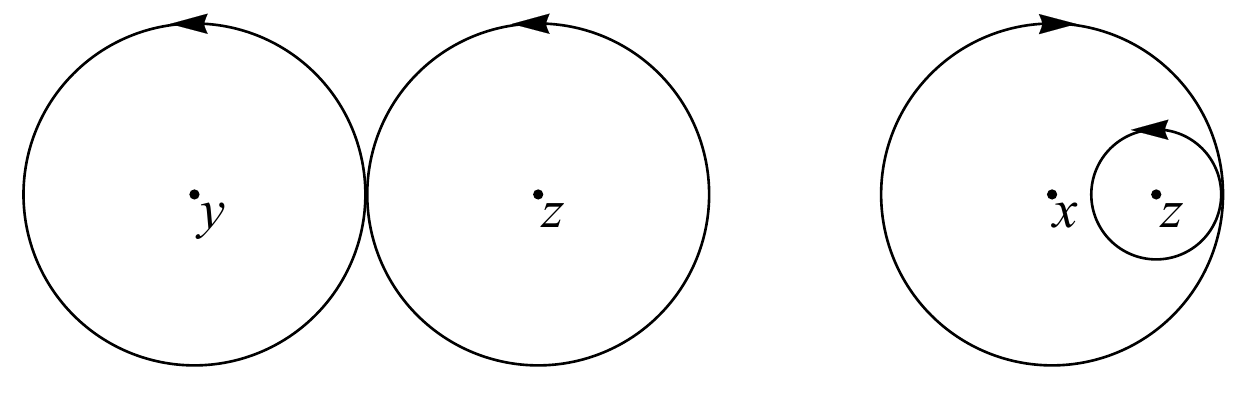}%
\caption{Transformation of the winding numbers under a change of puncture}%
\label{winding.fig}%
\end{figure}
%EndExpansion

\subsection{The span of the Makeenko--Migdal vectors\label{MMspan.sec}}

Let $L$ be a loop traced out on a graph in $S^{2}$ and having only simple
crossings. We consider vectors assigning real numbers to the faces of $L.$ For
each vertex $v$ of $L,$ we define the \textbf{Makeenko--Migdal vector}
associated to $v,$ denoted $\mathrm{MM}_{v},$ as
\[
\mathrm{MM}_{v}=\sum_{i=1}^{4}(-1)^{i+1}\delta_{F_{i}},
\]
where $F_{1},\ldots,F_{4}$ are the four faces surrounding $v.$ If, for
example, $F_{1},\ldots,F_{4}$ are distinct, then $\mathrm{MM}_{v}$ is the
vector that assigns the numbers $1,-1,1,-1$ to $F_{1},\ldots,F_{4},$
respectively, and zero to all other faces.

\begin{theorem}
[T. L\'{e}vy]\label{MMspan.thm}Fix a face $F_{0}$ of $L$ and let $u$ be a
vector assigning a real number to each face of $L.$ Then $u$ belongs to the
span of the Makeenko--Migdal vectors if and only if (1) $u$ is orthogonal to
the constant vector $\mathbf{1}:=(1,1,\ldots,1),$ and (2) $u$ is orthogonal to
the winding-number vector $w_{F_{0}}(\cdot).$
\end{theorem}

This result is the $r=1$ case of Lemma 9.4.3 in \cite{LevyMaster}. Note that
by Proposition \ref{windingIndep.prop}, the winding number vector
$w_{F_{0}^{\prime}}$ associated to some other face $F_{0}^{\prime}$ differs
from $w_{F_{0}}$ by a multiple of the constant vector $\mathbf{1}$. Thus, if
$u$ is orthogonal to $\mathbf{1},$ then $u$ is orthogonal to $w_{F_{0}}%
(\cdot)$ if and only if $u$ is orthogonal to $w_{F_{0}^{\prime}}(\cdot).$
Thus, the condition in the theorem is independent of the choice of $F_{0}.$

\subsection{Shrinking all but two of the faces\label{selecting.sec}}

Let $f$ denote the number of faces of $L.$ We now choose one face arbitrarily
and denote it by $F_{0}$. We will show that there is another face $F_{1}$ such
that we can perform a Makeenko--Migdal variation of the areas in which the
areas of the faces $F_{2},\ldots,F_{f-1}$ shrink simultaneously to zero, while
the area of $F_{0}$ \textit{increases} or remains the same. In the generic
case, the area of $F_{1}$ also remains positive, while in a certain borderline
case, the area of $F_{1}$ tends to zero as well.

The freedom we have to choose one of the \textquotedblleft
unshrunk\textquotedblright\ faces arbitrarily will be essential in
applications to the plane and to arbitrary surfaces. This freedom is a
difference between our approach and the one used in \cite{DN}, where the two
unshrunk faces in Section 4.5 are chosen to have maximal and minimal winding numbers.

\begin{proposition}
\label{shrinkAllButTwo.prop}Assume $L$ has at least three faces. Choose one
face $F_{0}$ arbitrarily and let all winding numbers be computed relative to a
puncture in $F_{0}$, so that $w(F_{0})=0.$ Let $\mathbf{a}=(a_{0},a_{1}%
,\ldots,a_{f-1})$ denote the vector of areas of the faces and let
$\mathbf{w}=(0,w_{1},\ldots,w_{f-1})$ be the vector of winding numbers.
Suppose $\mathbf{a}\cdot\mathbf{w}\neq0$ and adjust the labeling of
$F_{1},\ldots,F_{f-1}$ so that $w_{1}$ is maximal among the winding numbers if
$\mathbf{a}\cdot\mathbf{w}>0$ and $w_{1}$ is minimal if $\mathbf{a}%
\cdot\mathbf{w}<0$. Then there exists $\mathbf{b}$ in the span of the
Makeenko--Migdal vectors such that (1) all the entries of $\mathbf{a}%
+t\mathbf{b}$ are non-negative for $0\leq t\leq1$ and (2) $\mathbf{a}^{\prime
}:=\mathbf{a}+\mathbf{b}$ has the form $\mathbf{a}^{\prime}=(a_{0}^{\prime
},a_{1}^{\prime},0,\ldots,0)$ with $a_{0}^{\prime}\geq a_{0}$ and
$a_{1}^{\prime}>0$.

Meanwhile, if $\mathbf{a}\cdot\mathbf{w}=0$, there exists $\mathbf{b}$ in the
span of the Makeenko--Migdal vectors such that (1) all the entries of
$\mathbf{a}+t\mathbf{b}$ are non-negative for $0\leq t\leq1$ and (2)
$\mathbf{a}^{\prime}:=\mathbf{a}+\mathbf{b}$ has the form $\mathbf{a}^{\prime
}=(a_{0}^{\prime},0,\ldots,0)$ with $a_{0}^{\prime}>0$.
\end{proposition}

The proposition says, briefly, that we can shrink $a_{2},\ldots,a_{f-1}$ to
zero without decreasing $a_{0}$ and while keeping $a_{1}$ non-negative.

In the case of the trefoil loop, for example, suppose we take $F_{0}$ to be
the \textquotedblleft unbounded\textquotedblright\ face in Figure
\ref{trefoilshrink.fig} and we orient the loop in the counter-clockwise
direction. Then the winding numbers are 1 for the three lobes of the trefoil
and 2 for the central region. Since all winding numbers are positive in this
case, we always have $\mathbf{a}\cdot\mathbf{w}>0,$ in which case we take
$F_{1}$ to the central region. Figures \ref{trefoilshrink.fig} and
\ref{trefoilreduced.fig} then illustrate Proposition
\ref{shrinkAllButTwo.prop}.

\begin{proof}
Assume first that $\mathbf{a}\cdot\mathbf{w}>0$, in which case, the maximal
winding number $w_{1}$ must be positive. In light of Theorem \ref{MMspan.thm},
two vectors $\mathbf{a}$ and $\mathbf{a}^{\prime}$ differ by a vector in the
span of the Makeenko--Migdal vectors if $\mathbf{a}$ and $\mathbf{a}^{\prime}$
have the same inner products with the constant vector $\mathbf{1}$ and the
winding number vector $\mathbf{w}$. To achieve these conditions, we first set
$a_{1}^{\prime}$ equal to $\mathbf{a}\cdot\mathbf{w}/w_{1}$ and $a_{2}%
^{\prime},\ldots,a_{f-1}^{\prime}$ equal to 0. Since $w_{0}=0$, we then have
$\mathbf{a}^{\prime}\cdot\mathbf{w}=\mathbf{a}\cdot\mathbf{w}$, regardless of
the value of $a_{0}^{\prime}$. Next, we choose $a_{0}^{\prime}=-a_{1}^{\prime
}+a_{0}+a_{1}+\cdots+a_{f-1}$ to achieve the condition $\mathbf{a}^{\prime
}\cdot\mathbf{1}=\mathbf{a}\cdot\mathbf{1}$. Now, $\mathbf{a}^{\prime}%
\cdot\mathbf{w}=a_{1}w_{1}=\mathbf{a}\cdot\mathbf{w}$. Then since $w_{1}$ is
maximal, we have
\begin{align*}
a_{1}^{\prime}w_{1}  &  =a_{1}w_{1}+\cdots+a_{f-1}w_{f-1}\\
&  \leq w_{1}(a_{1}+\cdots+a_{f-1})
\end{align*}
Thus, $a_{1}^{\prime}\leq a_{1}+\cdots+a_{f-1}$, which shows that
$a_{0}^{\prime}\geq a_{0}$. (In particular, $a_{0}^{\prime}>0$.) Since both
$\mathbf{a}$ and $\mathbf{a}^{\prime}$ have non-negative entries, so does
every point along the line segment joining them. The analysis of the case
$\mathbf{a}\cdot\mathbf{w}<0$ is similar.

Finally, if $\mathbf{a}\cdot\mathbf{w}=0$, we may set $a_{0}^{\prime}%
=a_{0}+a_{1}+\cdots+a_{f-1}$ and all other entries of $\mathbf{a}^{\prime}$
equal to zero. In that case, $\mathbf{a}^{\prime}\cdot\mathbf{w}%
=\mathbf{a}\cdot\mathbf{w}=0$ and $\mathbf{a}^{\prime}\cdot\mathbf{1}%
=a_{0}^{\prime}=\mathbf{a}\cdot\mathbf{1}$, showing that $\mathbf{a}^{\prime}$
differs from $\mathbf{a}$ by a vector in the span of the Makeenko--Migdal vectors.
\end{proof}

\subsection{Analyzing the limiting case\label{limiting.sec}}

As a consequence of Proposition \ref{shrinkAllButTwo.prop}, we may start with
an arbitrary loop $L$ and perform a linear combination of Makeenko--Migdal
variations at each vertex, obtaining a family $L(t),$ $0\leq t<1,$ of loops
with the same topological type with all but two of the areas shrinking to zero
as $t\rightarrow1.$ We now analyze the behavior of the Wilson loop functional
in the limit $t\rightarrow1.$

We will analyze the limit \textquotedblleft analytically,\textquotedblright%
\ using Sengupta's formula for the finite-$N$ case. (Recall Remark
\ref{varyAreas.remark}.) For this result, the structure group can be an
arbitrary connected compact Lie group $K.$

\begin{theorem}
\label{shrinkingLim.thm}Let $L$ be a loop traced out on a graph in $S^{2}$ and
having only simple crossings. Denote the number of faces of $L$ by $f$ and
label the faces as $F_{0},F_{1},F_{2},\ldots,F_{f-1}.$ Suppose we vary the
areas of the faces as a function of a parameter $t\in\lbrack0,1)$ in such a
way that as $t\rightarrow1,$ the areas of $F_{2},\ldots,F_{f-1}$ tend to zero,
while the areas of $F_{0}$ and $F_{1}$ approach positive real numbers $a$ and
$c$, respectively. Then%
\begin{equation}
\lim_{t\rightarrow1}\mathbb{E}\{\mathrm{tr}(\mathrm{hol}(L))\}=\frac{1}{Z}%
\int_{K}\mathrm{tr}(h^{n})\rho_{a}(h)\rho_{c}(h)~dh, \label{limitingWL}%
\end{equation}
where $n=w_{F_{0}}(F_{1})$ is the winding number of $L$ around $F_{1},$
relative to a puncture in $F_{0},$ and $Z$ is a normalization constant.
Meanwhile, if the area of $F_{1}$ also tends to zero in the $t\rightarrow1$
limit, then
\[
\lim_{t\rightarrow1}\mathbb{E}\{\mathrm{tr}(\mathrm{hol}(L))\}=1.
\]

\end{theorem}

The integral on the right-hand side of (\ref{limitingWL}) is just the Wilson
loop functional for a loop $L_{n}$ that winds $n$ times around a simple closed
curve enclosing areas $a$ and $c.$ Note that the winding number of $L_{n}$
around $F_{0}$ is the same as the winding number of the original loop $L$
around $F_{0}.$ Note also that by Theorem \ref{MMspan.thm}, the value of
$\mathbf{a}\cdot\mathbf{w}$ for the loop $L_{t}$ is independent of $t$ for
$t\in\lbrack0,1).$ Since $w_{0}=0$ and $a_{2}$ through $a_{f-1}$ are tending
to zero, this means that $a_{1}$ must be tending to the value $\mathbf{a}%
\cdot\mathbf{w}/n,$ where $n=w_{1}$ is the winding number of $L$ around
$F_{1}.$ It then follows that the limiting loop $L_{n}$ has the same value of
$\mathbf{a}\cdot\mathbf{w},$ even though $L_{n}$ has a different topological
type from $L$.

\begin{proof}
Let $\mathbb{G}$ be a minimal oriented graph on which $L$ can be traced. We
think of $\mathbb{G}$ as a graph in the plane, by placing a puncture into
$F_{0}.$ If $e$ denotes the number of edges of $\mathbb{G},$ then we may
consider two different measures on $K^{e}$: the Yang--Mills measure
$\mu_{\mathrm{plane}}^{\mathbb{G}}$ for $\mathbb{G}$ viewed as a graph in the
plane and the Yang--Mills measure $\mu_{\mathrm{sphere}}^{\mathbb{G}}$ for
$\mathbb{G}$ viewed as a graph in the sphere. By comparing Sengupta's formula
\cite[Sect. 5]{Sen93} in the sphere case to Driver's formula \cite[Theorem
6.4]{Dr} in the plane case, we see that%
\[
d\mu_{\mathrm{sphere}}^{\mathbb{G}}(\mathbf{x})=\frac{1}{Z}\rho_{\left\vert
F_{0}\right\vert }(\mathrm{hol}_{F_{0}}(\mathbf{x}))~d\mu_{\mathrm{plane}%
}^{\mathbb{G}}(\mathbf{x}),
\]
where $\mathbf{x}\in K^{e}$ is the collection of edge variables, $\left\vert
F_{0}\right\vert $ is the area of $F_{0}$ as a face in $S^{2},$ and
$\mathrm{hol}_{F_{0}}$ is the product of edge variables around the boundary of
$F_{0}.$ Here $Z$ is a normalization constant.

We may rewrite both of the Yang--Mills measures using \textquotedblleft loop
variables\textquotedblright\ as follows. Let us fix a vertex $v$ and a
spanning tree $T$ for $\mathbb{G}.$ Then in Section 5.3 of \cite{LevyMaster},
L\'{e}vy gives a procedure associating to the faces $F_{1},\ldots,F_{f-1}$
certain loops $L_{1},\ldots,L_{f-1}$ in $\mathbb{G}$ that constitute free
generators for $\pi_{1}(\mathbb{G})$, based at $v.$ (Note that there is no
generator associated to the face $F_{0}.$) Each $L_{i}$ is a word in the edges
of $\mathbb{G}.$ We may then associate to each of the faces $F_{i},$
$i=1,\ldots,f-1,$ of $\mathbb{G}$ a loop \textit{variable}, which is the
product of edge variables in $L_{i}$ (in the reverse order, since parallel
transport reverses order). The map sending the collection of edge variables to
the collection of loop variables defines a map $\Gamma:K^{e}\rightarrow
K^{f-1},$ where $e$ is the number of edges of $\mathbb{G}.$

According to Proposition 5.3.3 of \cite{LevyMaster}, the loop variables are
independent heat-kernel distributed random variables with respect to
$\mu_{\mathrm{plane}}^{\mathbb{G}}.$ (See also Proposition 12.7 of \cite{Ga}.)
That is to say, the push-forward of $\mu_{\mathrm{plane}}^{\mathbb{G}}$ under
$\Gamma$ is just the product of heat kernels, at times equal to the areas of
the faces. Although there is no generator associated to the face $F_{0},$ the
$L_{1},\ldots,L_{f-1}$ generate $\pi_{1}(\mathbb{G}).$ Suppose, therefore,
that $L_{0}$ is a loop that starts at $v,$ travels along a path $P$ to a
vertex in $\partial F_{0},$ then around the boundary of $F_{0},$ and then back
to $v$ along $P^{-1}.$ Then $L_{0}$ is expressible as a word in $L_{1}%
,\ldots,L_{f-1}$ and therefore $\mathrm{hol}_{F_{0}}(\mathbf{x})$ is
expressible as a word in the loop variables: $\mathrm{hol}_{F_{0}}%
(\mathbf{x})=w_{0}(\Gamma(\mathbf{x}))$ for some function $w_{0}$ on
$K^{f-1}.$ (Although $\mathrm{hol}_{F_{0}}(\mathbf{x})$ depends on the choice
of $v$ and $P,$ the invariance of the heat kernel under conjugation guarantees
that $\rho_{\left\vert F_{0}\right\vert }(\mathrm{hol}_{F_{0}}(\mathbf{x}))$
is well defined.) It then follows from the measure-theoretic change of
variables theorem that the push-forward of $\mu_{\mathrm{sphere}}^{\mathbb{G}%
}$ is given by%
\begin{align*}
&  d\Gamma_{\ast}(\mu_{\mathrm{sphere}}^{\mathbb{G}})(h_{1},\ldots,h_{f-1})\\
&  =\frac{1}{Z}\rho_{\left\vert F_{0}\right\vert }(w_{0}(h_{1},\ldots
,h_{f-1}))\left(  \prod_{i=1}^{f-1}\rho_{\left\vert F_{i}\right\vert }%
(h_{i})\right)  dh_{1}~\ldots dh_{f-1}.
\end{align*}
%

%TCIMACRO{\FRAME{ftbpFU}{4.2281in}{1.9009in}{0pt}{\Qcb{The example loop $L$
%(left) decomposes as $L_{1}L_{2}L_{3}L_{1}^{-1}L_{4}^{-1}L_{3}^{-1}$}%
%}{\Qlb{decompose.fig}}{decompose.eps}{\special{ language "Scientific Word";
%type "GRAPHIC";  maintain-aspect-ratio TRUE;  display "USEDEF";
%valid_file "F";  width 4.2281in;  height 1.9009in;  depth 0pt;
%original-width 3.5008in;  original-height 1.5618in;  cropleft "0";
%filename 'decompose.eps';file-properties "XNPEU";}} }%
%BeginExpansion
\begin{figure}[ptb]%
\centering
\includegraphics[
height=1.9009in,
width=4.2281in
]%
{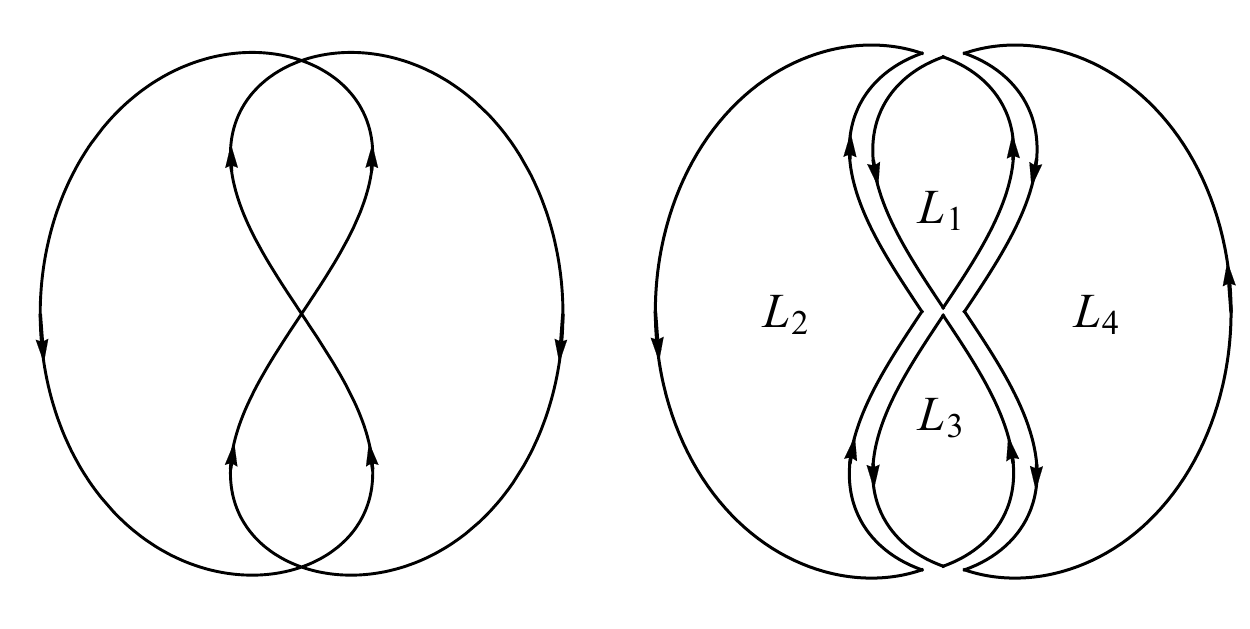}%
\caption{The example loop $L$ (left) decomposes as $L_{1}L_{2}L_{3}L_{1}%
^{-1}L_{4}^{-1}L_{3}^{-1}$}%
\label{decompose.fig}%
\end{figure}
%EndExpansion

Meanwhile, the loop $L$ (whose Wilson loop functional we are considering) is
also expressible as a word $w_{1}$ in the generators $L_{1},\ldots,L_{f-1}.$
(See Figure \ref{decompose.fig}.) Thus,
\begin{align}
&  \int_{K^{e}}\mathrm{tr}(\mathrm{hol}(L))~d\mu_{\mathrm{sphere}}%
^{\mathbb{G}}=\frac{1}{Z}\int_{K^{f-1}}\mathrm{tr}(w_{1}(h_{1},\ldots
,h_{f-1}))\nonumber\\
&  \times\rho_{\left\vert F_{0}\right\vert }(w_{0}(h_{1},\ldots,h_{f-1}%
))\left(  \prod_{i=1}^{f-1}\rho_{\left\vert F_{i}\right\vert }(h_{i})\right)
dh_{1}~\ldots dh_{f-1}. \label{loopVarExpectation}%
\end{align}
It is now straightforward to take a limit in (\ref{loopVarExpectation}) as
$t\rightarrow1,$ that is, as all areas other than $\left\vert F_{0}\right\vert
$ and $\left\vert F_{1}\right\vert $ tend to zero. In this limit, each heat
kernel associated to $h_{i},$ $i\geq2,$ becomes a $\delta$-measure, so we
simply evaluate each such $h_{i}$ to the identity element of $K,$ giving%
\begin{align}
&  \lim_{t\rightarrow1}\int_{K^{e}}\mathrm{tr}(\mathrm{hol}(L))~d\mu
_{\mathrm{sphere}}^{\mathbb{G}}\nonumber\\
&  =\frac{1}{Z}\int_{K}\mathrm{tr}(w_{1}(h_{1},\mathrm{id},\ldots
,\mathrm{id}))\rho_{a}(w_{0}(h_{1},\mathrm{id},\ldots,\mathrm{id}))\rho
_{c}(h_{1})~dh_{1}. \label{limSigmaOne}%
\end{align}
In the sphere case, the normalization factor is given by $Z=\rho
_{A}(\mathrm{id}),$ where $A$ is the total area of the sphere. Although $Z$
may depend on $t$ (since we do not currently assume that the area of the
sphere is fixed), it has a limit as $t\rightarrow1.$

If the limiting value $a$ of $\left\vert F_{1}\right\vert $ is also zero, then
$\rho_{\left\vert F_{1}\right\vert }$ becomes also becomes a $\delta
$-function. Since the normalized trace of the identity matrix equals 1, the
right-hand side of (\ref{limSigmaOne}) becomes $\rho_{c}(\mathrm{id})/Z.$ But
the total area of the sphere in this limit is just $c,$ so $Z=\rho
_{c}(\mathrm{id})$ and $\rho_{c}(\mathrm{id})/Z=1.$

Finally, if the limiting value~$a$ of $\left\vert F_{1}\right\vert $ is
nonzero, we must understand the effect on $w_{0}$ and $w_{1}$ of evaluating
$h_{i},~i\geq2,$ to the identity. Recall that the words $w_{0}$ and $w_{1}$
arise from representing the boundary of $F_{0}$ and the loop $L$ as words in
$L_{1},\ldots,L_{f-1},$ which are free generators for $\pi_{1}(\mathbb{G})$.
Now, $\pi_{1}(\mathbb{G})$ is naturally isomorphic to $\pi_{1}(\mathbb{R}%
^{2}\setminus\{x_{1},\ldots,x_{f-1}\}),$ where $x_{i}$ is an arbitrarily
chosen element of $F_{i}.$ There is then a homomorphism from $\pi
_{1}(\mathbb{G})$ to $\pi_{1}(\mathbb{R}^{2}\setminus\{x_{1}\})$ induced by
the inclusion of $\mathbb{R}^{2}\setminus\{x_{1},\ldots,x_{f-1}\}$ into
$\mathbb{R}^{2}\setminus\{x_{1}\}.$ Since $\pi_{1}(\mathbb{R}^{2}%
\setminus\{x_{1}\})$ is just a free group on the single generator $L_{1},$
this homomorphism is computed by mapping each of the generators $L_{2}%
,\ldots,L_{f-1}$ to the identity, leaving only powers of $L_{1}.$ Thus, if we
write, say, $\partial F_{0}$ as a word in the generators $L_{1},\ldots
,L_{f-1}$ and apply the just-mentioned homomorphism, the result will be
$L_{1}^{n_{0}},$ where $n_{0}$ is the winding number of $\partial F_{0}$
around $x_{1}$ (i.e., around $F_{1}$). By the Jordan curve theorem, this
winding number is $1,$ assuming that we traverse $\partial F_{0}$ in the
counter-clockwise direction. (In Figure \ref{exampleouter.fig}, for example,
the outer boundary of the loop in Figure \ref{decompose.fig} decomposes as
$L_{2}L_{3}L_{4}L_{1}$; if we set any three of the four generators to the
identity, the remaining generator will occur to the power 1.) Expressing this
result in terms of the loop \textit{variables}, rather than the loops
themselves, we conclude that $w_{0}(h_{1},\mathrm{id},\ldots,\mathrm{id})$ is
$h_{1}$. Similarly, $w_{1}(h_{1},\mathrm{id},\ldots,\mathrm{id})=h_{1}^{n},$
where $n$ is the winding number of $L$ around $F_{1}.$
\end{proof}

%

%TCIMACRO{\FRAME{ftbpFU}{1.9873in}{1.9839in}{0pt}{\Qcb{The outer boundary of
%the loop in Figure \ref{decompose.fig} is computed as $L_{2}L_{3}L_{4}L_{1}$}%
%}{\Qlb{exampleouter.fig}}{exampleouter.eps}%
%width 1.9873in;  height 1.9839in;  depth 0pt;  original-width 3.2655in;
%original-height 3.2603in;  cropleft "0";  croptop "1";  cropright "1";
%cropbottom "0";  filename 'exampleouter.eps';file-properties "XNPEU";}} }%
%BeginExpansion
\begin{figure}[ptb]%
\centering
\includegraphics[
height=1.9839in,
width=1.9873in
]%
{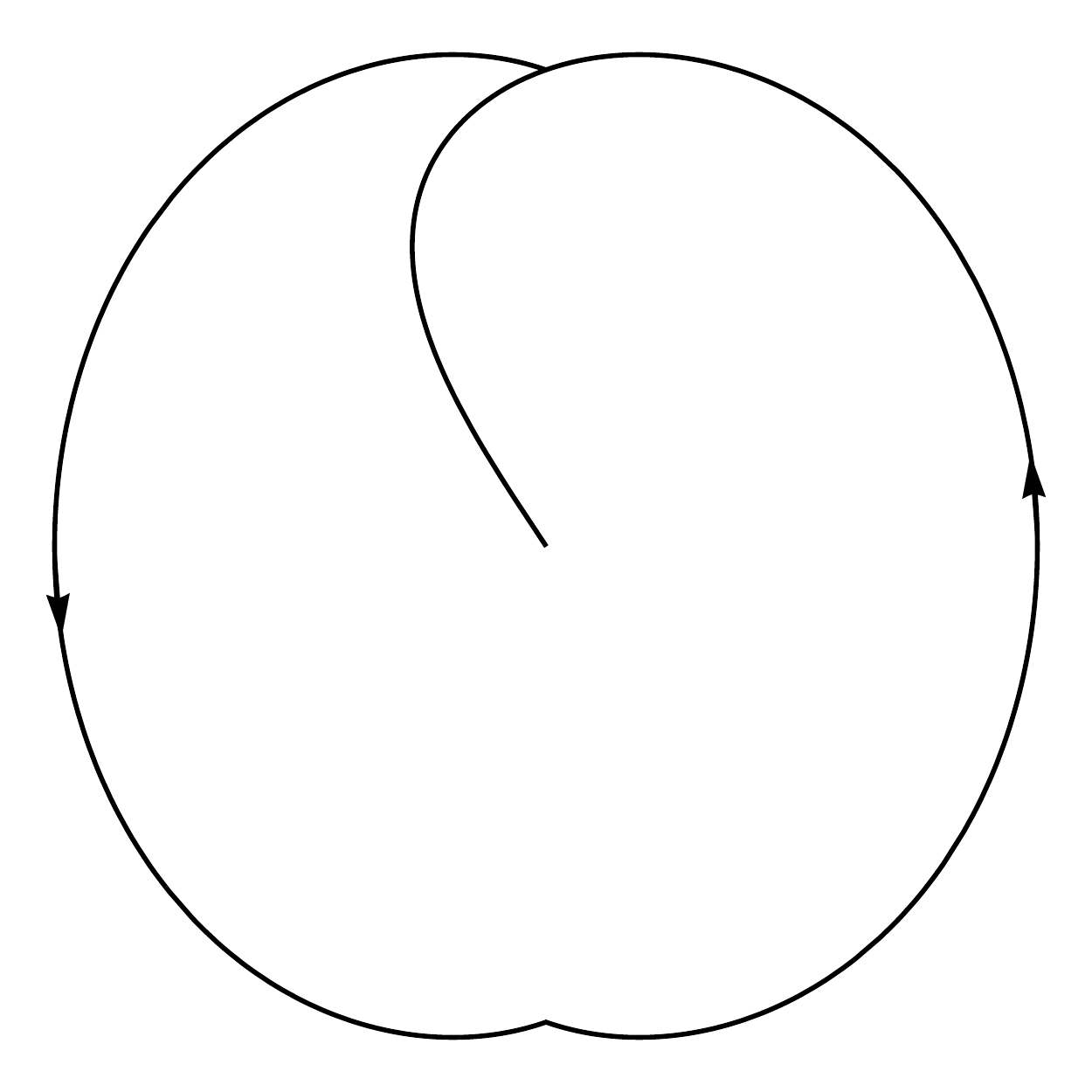}%
\caption{The outer boundary of the loop in Figure \ref{decompose.fig} is
computed as $L_{2}L_{3}L_{4}L_{1}$}%
\label{exampleouter.fig}%
\end{figure}
%EndExpansion

\subsection{The induction\label{induction.sec}}

In this section, we prove Theorem \ref{main.thm} by induction on the number of
crossings. Our strategy is to deform an arbitrary loop $L$ with $k$ crossings
into a loop $L_{n}$ that winds $n$ times around a simple closed curve. The
variation of the Wilson loop functional along this deformation will be a
linear combination of products of Wilson loop functionals for curves with
fewer crossings, plus a covariance term. We begin by recording a result for a
loop that winds $n$ times around a simple closed curve.

\begin{theorem}
\label{nfold.thm}Let $L_{n}(a,c)$ denote a loop that winds $n$ times around a
simple closed curve enclosing areas $a$ and $c.$ Then for all $a$ and $c,$ the
limit%
\[
W_{n}(a,c):=\mathbb{E}\left\{  \mathrm{tr}(\mathrm{hol}(L_{n}(a,c)))\right\}
\]
exists and depends continuously on $a$ and $c.$ Furthermore, the associated
variance tends to zero:%
\[
\lim_{N\rightarrow\infty}\mathrm{Var}\left\{  \mathrm{tr}(\mathrm{hol}%
(L_{n}(a,c)))\right\}  =0.
\]

\end{theorem}

The proof of this result is given in Section \ref{nfold.sec}. Assuming Theorem
\ref{nfold.thm} for the moment, we are now ready for the proof of our main result.

\begin{proof}
[Proof of Theorem \ref{main.thm}]Let $k$ denote the number of crossings of
$L.$ We will proceed by induction on $k.$ When $k=0,$ the result is precisely
Theorem \ref{dn.thm}. Assume, then, that (\ref{main1}) and (\ref{main2}) and
the continuity condition hold for all loops with $l<k$ crossings and consider
a loop $L$ with $k$ crossings. Using Proposition \ref{shrinkAllButTwo.prop},
we may make a combination of Makeenko--Migdal variations at the vertices of
$L,$ giving loops $L(t)$ in which $L(0)=L$ and so that all but two of the
faces shrink to zero as $t\rightarrow1.$ By Theorem \ref{shrinkingLim.thm},
the limit as $t\rightarrow1$ of the Wilson loop functional of $L(t)$ is the
Wilson loop functional of a loop $L^{n}(a,c)$ that winds $n$ times around a
simple closed curve enclosing areas $a$ and $c,$ where $a$ and $c$ are the
areas of the two remaining faces. We may differentiate the Wilson loop
functional of $L(t)$ using the chain rule and the first part of Proposition
\ref{variation.prop}. Integrating the derivative then gives
\begin{align}
\mathbb{E}\left\{  \mathrm{tr}(\mathrm{hol}(L))\right\}   &  =\mathbb{E}%
\left\{  \mathrm{tr}(\mathrm{hol}(L_{n}(a,c)))\right\} \nonumber\\
&  -\sum_{j}c_{j}\int_{0}^{1}\mathbb{E}\left\{  \mathrm{tr}(\mathrm{hol}%
(L_{1,j}(t)))\right\}  \mathbb{E}\left\{  \mathrm{tr}(\mathrm{hol}%
(L_{2,j}(t)))\right\}  ~dt\nonumber\\
&  -\sum_{j}c_{j}\mathrm{Cov}\{\mathrm{tr}(\mathrm{hol}(L_{1,j}%
(t))),\mathrm{tr}(\mathrm{hol}(L_{2,j}(t)))\}~dt, \label{inductionIntegral}%
\end{align}
where the constants $c_{j}$ are the ones expressing the vector $\mathbf{b}$ in
Proposition \ref{shrinkAllButTwo.prop} as a linear combination of the
Makeenko--Migdal vectors $\mathrm{MM}_{v_{j}}.$ Here $L_{1,j}(t)$ and
$L_{2,j}(t)$ are the loops obtained by applying the Makeenko--Migdal equation
at the vertex $v_{j}$ to the loop $L(t).$ In particular, since neither of
these loops has a crossing at $v_{j},$ we see that $L_{1,j}(t)$ and
$L_{2,j}(t)$ have at most $k-1$ crossings.

By our induction hypothesis, the limit%
\[
W(L_{i,j}(t)):=\lim_{N\rightarrow\infty}\mathbb{E}\left\{  \mathrm{tr}%
(\mathrm{hol}(L_{i,j}(t)))\right\}
\]
exists for each $t\in\lbrack0,1),$ each $i\in\{1,2\},$ and each $j.$ Our
induction hypothesis also tells us that the variance of $\mathrm{tr}%
(\mathrm{hol}(L_{i,j}(t)))$ goes to zero as $N\rightarrow\infty$; the
inequality (\ref{covIneq}) then tells us that the covariances on the last line
of (\ref{inductionIntegral}) also tend to zero. Now, $\left\vert
\mathrm{tr}(\mathrm{hol}(U))\right\vert \leq1$ for all $U\in U(N),$ from which
it follows that $\left\vert \mathrm{Var}(\mathrm{tr}(\mathrm{hol}%
(U)))\right\vert \leq1$. Thus, we may apply dominated convergence to all
integrals in (\ref{inductionIntegral}), along with Theorem \ref{nfold.thm}, to
obtain%
\begin{align}
\lim_{N\rightarrow\infty}\mathbb{E}\left\{  \mathrm{tr}(\mathrm{hol}%
(L))\right\}   &  =W(L_{n}(a,c))\nonumber\\
&  -\sum_{j}c_{j}\int_{0}^{1}W(L_{1,j}(t))W(L_{2,j}(t))~dt,
\label{inductionLim}%
\end{align}
establishing the existence of the limit in (\ref{main1}) of Theorem
\ref{main.thm}.

We now establish the claimed continuity of $W(L)$ with respect to the areas of
the faces. For each topological type of loop $L,$ we fix an arbitrary face to
play the role of $F_{0}$ in the proof of Proposition
\ref{shrinkAllButTwo.prop}. We then fix two other faces with maximal and
minimal winding numbers, respectively, relative to a puncture in $F_{0}$. The
first key observation is that with these choices made, the vector
$\mathbf{a}^{\prime}$ in the proof of the proposition can be chosen to depend
continuously on the areas of the faces. If we compute $\mathbf{a}^{\prime}$ by
the formula in the proof, the only case in which continuity is not obvious is
in the case in which (in the notation of the proof) $\mathbf{a}\cdot
\mathbf{w}=0.$ But a moment's thought will confirm that as $\mathbf{a}%
\cdot\mathbf{w}$ tends to zero, $\mathbf{a}^{\prime}$ tends to the vector
$(a_{0},0,\ldots,0),$ where $a_{0}$ is the area of $S^{2},$ which is just the
value of $\mathbf{a}^{\prime}$ when $\mathbf{a}\cdot\mathbf{w}=0.$

Now, when we deform our original loop $L$ into a loop of the form
$L_{n}(a,c),$ the values of $a$ and $c$ depend continuously on the areas of
the faces of $L$; indeed, $c$ and $a$ are the first two entries of the vector
$\mathbf{a}^{\prime}.$ The continuity of $\mathbf{a}^{\prime}$ also gives the
continuity of $\mathbf{b}=\mathbf{a}^{\prime}-\mathbf{a}$ and thus of the
coefficients $c_{j}$ in (\ref{inductionLim}), which are the coefficients of
the expansion of $\mathbf{b}$ in terms of the Makeenko--Migdal vectors. Thus,
the continuity of the second term on the right-hand side of
(\ref{inductionLim}) follows from our induction hypothesis and dominated convergence.

To establish the second claim (\ref{main2}) of Theorem \ref{main.thm}, we use
the second part of Proposition \ref{variation.prop}. We then need to bound the
covariance term appearing in (\ref{variationVariance}). We first use
(\ref{covIneq}) and then use Proposition \ref{varianceProd.prop} to bound the
variance of the product of $\mathrm{tr}(\mathrm{hol}(L_{1}))$ and
$\mathrm{tr}(\mathrm{hol}(L_{2})).$ The argument is then similar to the proof
of (\ref{main1}).

Finally, we establish the large-$N$ Makeenko--Migdal equation (\ref{main3})
for the limiting Wilson loop functionals. If we vary the areas in a
checkerboard pattern along a path $L(t)$ as in Figure \ref{lt.fig}, we have%
\begin{align*}
\mathbb{E}\left\{  \mathrm{tr}(\mathrm{hol}(L(t)))\right\}   &  =\mathbb{E}%
\left\{  \mathrm{tr}(\mathrm{hol}(L(t_{0})))\right\} \\
&  +\int_{t_{0}}^{t}\mathbb{E}\left\{  \mathrm{tr}(\mathrm{hol}(L_{1}%
(s)))\right\}  \mathbb{E}\left\{  \mathrm{tr}(\mathrm{hol}(L_{2}(s)))\right\}
~ds\\
&  +\int_{t_{0}}^{t}\mathrm{Cov}\{\mathrm{tr}(\mathrm{hol}(L_{1}%
(s))),\mathrm{tr}(\mathrm{hol}(L_{2}(s)))\}~ds.
\end{align*}
Using the first two points (\ref{main1})\ and (\ref{main2}) in Theorem
\ref{main.thm}, we can let $N$ tend to infinity to obtain%
\begin{equation}
W(L(t))=W(L(t_{0}))+\int_{t_{0}}^{t}W(L_{1}(s))W(L_{2}(s))~ds.
\label{generalMM}%
\end{equation}
Since we have shown that $W(L)$ depends continuously on the areas of $L,$ we
see that $W(L_{1}(s))$ and $W(L_{2}(s))$ depend continuously on $s.$ Thus, we
can apply the fundamental theorem of calculus to differentiate
(\ref{generalMM}) with respect to $t$ to obtain the large-$N$ Makeenko--Migdal
equation (\ref{main3}).
\end{proof}

\subsection{The $n$-fold circle\label{nfold.sec}}

In this section, we analyze the Wilson loop functional for a loop that winds
$n$ times around a simple closed curve, establishing Theorem \ref{nfold.thm}.
Similar results were obtained by Dahlqvist and Norris in \cite[Section
4.5]{DN}.

Recall from (\ref{simpleMeasure}) that the distribution of the parallel
transport around a simple closed curve is the distribution of a Brownian
bridge in $U(N).$ Assuming that the eigenvalue distribution of a Brownian
bridge in $U(N)$ has a deterministic large-$N$ limit, the calculations here
provide a recursive method for computing the higher moments of the limiting
distribution, in terms of the first moment. (But the formula for the first
moment is complicated, especially in the large-lifetime phase, $a+b>\pi^{2}.$)
A similar, but not identical, recursion for the moments of the free unitary
Brownian motion---i.e., the large-$N$ limit of a Brownian \textit{motion} in
$U(N),$ rather than a Brownian bridge---was established by Biane. See the
bottom of p. 16 of \cite{Bi1} and p. 266 of \cite{Bi2}.

We approximate a loop that winds $n$ times around a simple closed curve by a
loop with simple crossings, specifically, by a curve of the \textquotedblleft
loop within a loop within a loop\textquotedblright\ form, as in Figure
\ref{5fold.fig}. If there are a total of $n$ loops, then there are $n-1$
crossings and a total of $n+1$ faces. The faces divide into the innermost
circle, the outer region, and $n-1$ annular regions. Shifting the puncture
from the innermost to outermost regions gives another curve of the same sort,
with the order of the annular regions reversed.

Now, nothing prevents us from letting the areas of one or more of the annular
regions equal zero. Specifically, we may understand the limit as some of the
areas tend to zero by the same method as in the proof of Theorem
\ref{shrinkingLim.thm}. The only difference is that in the expression
(\ref{loopVarExpectation}), we set an arbitrary subset of the loop variables
equal to the identity. We may similarly make a Makeenko--Migdal variation at
one or more of the vertices and then take a limit on both sides of the
resulting equation as some of the areas tend to zero.

If we set the areas of the $n-1$ annular regions to zero, the curve will
simply wind $n$ times around the outermost circle. More generally, we will
show that it is possible to shrink the central region and the outer region,
while increasing the innermost annular region and keeping the areas of the
outermost $n-2$ annular regions equal to zero, as in Figure \ref{5fold.fig}.
This observation motivates consideration of the following class of loops.
%TCIMACRO{\FRAME{ftbpFU}{2.9213in}{2.3514in}{0pt}{\Qcb{The loop $L_{5}(a,b,c)$%
%}}{\Qlb{5fold.fig}}{5fold.eps}{\special{ language "Scientific Word";
%type "GRAPHIC";  maintain-aspect-ratio TRUE;  display "USEDEF";
%valid_file "F";  width 2.9213in;  height 2.3514in;  depth 0pt;
%original-width 3.4316in;  original-height 2.8591in;  cropleft "0";
%filename '5fold.eps';file-properties "XNPEU";}} }%
%BeginExpansion
\begin{figure}[ptb]%
\centering
\includegraphics[
height=2.3514in,
width=2.9213in
]%
{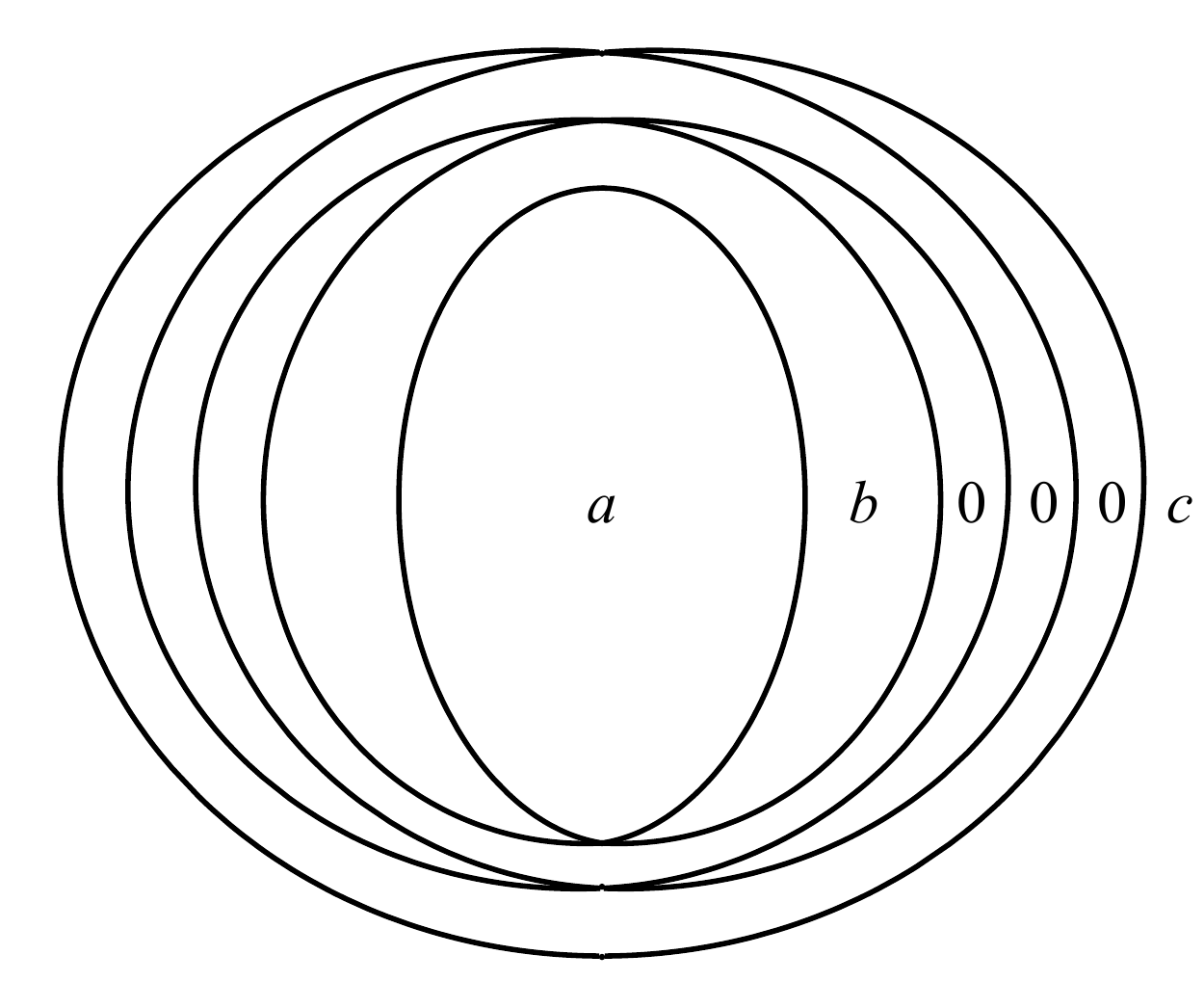}%
\caption{The loop $L_{5}(a,b,c)$}%
\label{5fold.fig}%
\end{figure}
%EndExpansion

\begin{definition}
Let $L_{n}(a,c)$ denote the loop that winds $n$ times around a simple closed
curve enclosing areas $a$ and $c.$ Now let $L_{1}(a+b,c)$ be a simple closed
curve having two faces $F_{1}$ and $F_{2}$ with areas $a+b$ and $c,$
respectively. For $n\geq2,$ let $L_{n}(a,b,c)$ denote the loop that winds
$n-1$ times around $L_{1}(a+b,c)$ and then winds once around a region of area
$a$ inside $F_{1}.$
\end{definition}

When $n=1,$ we interpret $L_{1}(a,b,c)$ as being $L_{1}(a,b+c)$ (the loop that
winds zero times around around $L_{1}(a+b,c)$ and then once around a region of
area $a$).

We may consider two different limiting cases of $L_{n}(a,b,c),$ the case where
$a$ tends to zero and the case where $c$ tends to zero. If $a$ tends to zero,
we are left with $L_{n-1}(b,c),$ while if $c$ tends to zero, we are left with
$L_{1}(a,b).$ (To see this, it may be useful to shift the puncture from the
$c$ face to the $b$ face, as on the right-hand side of Figure \ref{double.fig}%
.) We now give an inductive procedure for computing the limiting Wilson loop
functionals for loops of the form $L_{n}(a,c)$ and $L_{n}(a,b,c).$

\begin{theorem}
\label{WnPlus1.thm}For all $n\geq1$ and all non-negative real numbers $a,$
$b,$ and $c,$ the limits
\begin{equation}
W_{n}(a,c):=\lim_{N\rightarrow\infty}\mathbb{E}\left\{  \mathrm{tr}%
(\mathrm{hol}(L_{n}(a,c)))\right\}  \label{wnLim}%
\end{equation}
and%
\begin{equation}
W_{n}(a,b,c):=\lim_{N\rightarrow\infty}\mathbb{E}\left\{  \mathrm{tr}%
(\mathrm{hol}(L_{n}(a,b,c)))\right\}  \label{WnLim}%
\end{equation}
exist, and the associated variances tend to zero. Furthermore, for all
$n\geq1,$ we have the following inductive formulas. If $c\geq a/n,$ we have%
\begin{align}
&  W_{n+1}(a,b,c)=W_{n}\left(  \frac{n+1}{n}a+b,c-\frac{a}{n}\right)
\nonumber\\
&  -\sum_{k=1}^{n}k\int_{0}^{a/n}W_{k}(a+b+t,c-t)W_{n+1-k}%
(a-nt,b+(n+1)t,c-t)~dt. \label{WnPlus1}%
\end{align}
If, on the other hand, $c<a/n,$ we have%
\begin{align}
&  W_{n+1}(a,b,c)=W_{1}\left(  a-nc,b+(n+1)c\right) \nonumber\\
&  -\sum_{k=1}^{n}k\int_{0}^{c}W_{k}(a+b+t,c-t)W_{n+1-k}%
(a-nt,b+(n+1)t,c-t)~dt. \label{WnPlus12}%
\end{align}

\end{theorem}

Our main interest is in the quantity $W_{n}(a,c),$ which is the same as
$W_{n}(a,0,c).$ Note, however, that even if we put $b=0$ on the left-hand side
of (\ref{WnPlus1}) or (\ref{WnPlus12}), the right-hand side of the equation
will still involve $W_{n+1-k}(a^{\prime},b^{\prime},c^{\prime})$ with nonzero
values of $b^{\prime}.$ On the other hand, the inductive procedure ultimately
expresses $W_{n}(a,b,c)$---and thus, in particular, $W_{n}(a,c)=W_{n}%
(a,0,c)$---in terms of $W_{1}(\cdot,\cdot).$

Suppose, for example, that we wish to compute $W_{3}(a,c)=W_{3}(a,0,c)$ in
terms of $W_{1}(\cdot,\cdot).$ Since $W_{3}(a,c)$ is symmetric in $a$ and $c,$
it is harmless to assume that $c\geq a,$ so that (\ref{WnPlus1}) applies. We
first compute $W_{2}(a,b,c)$ in terms of $W_{1}(\cdot,\cdot)$, by applying
(\ref{WnPlus1}) with $n=1$ (since $c\geq a$) giving%
\begin{align}
W_{2}(a,b,c)  &  =W_{1}(2a+b,c-a)\nonumber\\
&  -\int_{0}^{a}W_{1}(a+b+t,c-t)W_{1}(a-t,b+c+t)~dt,\quad c\geq a.
\label{w2induct}%
\end{align}
(This formula is just what we obtained in (\ref{figure8large}) in Section
\ref{figure8.sec}.) We then apply (\ref{WnPlus1}) with $n=2$ and $b=0,$ giving%
\begin{align}
W_{3}(a,c)  &  =W_{2}(3a/2,c-a/2)\nonumber\\
&  -\int_{0}^{a/2}W_{1}(a+t,c-t)W_{2}(a-2t,3t,c-t)~dt\nonumber\\
&  -2\int_{0}^{a/2}W_{2}(a+t,c-t)W_{1}(a-2t,c+2t)~dt. \label{w3induct}%
\end{align}
Last, we plug into (\ref{w3induct}) the values of $W_{2}(a-2t,3t,c-t)$ and
$W_{2}(a+t,c-t)=W_{2}(a+t,0,c-t)$ computed in (\ref{w2induct}).

In the case of $W_{2}(a-2t,3t,c-t),$ the assumption that $c\geq a$ guarantees
that $c-t\geq a-2t,$ so that we may directly apply (\ref{w2induct}). In the
case of $W_{2}(a+t,c-t)$, there may be some values of $t$ for which we need to
use the symmetry of $W_{2}(a,c)$ with respect to $a$ and $c$ before applying
(\ref{w2induct}). If, however, $c\geq2a,$ we may directly apply
(\ref{w2induct}) to the computation of $W_{2}(a+t,c-t)$ for all $0\leq t\leq
a/2.$

Some explicit computations of these expressions in the plane case (i.e., the
$c=\infty$ case) are given in Section \ref{planeExamples.sec}.

\begin{proof}
[Proof of Theorem \ref{WnPlus1.thm}]The argument is similar to the proof of
Theorem \ref{main.thm}, except with a different variation of the areas. We
assume, inductively, that the limits in (\ref{wnLim}) and (\ref{WnLim}) exist
for all $k\leq n$. We also assume that the corresponding variances go to zero.
(When $n=1,$ this claim is just Theorem \ref{dn.thm}.) We then prove these
results for $n+1,$ while also proving the inductive formulas (\ref{WnPlus1})
and (\ref{WnPlus12}).

We consider the loop $L_{n}(a,b,c).$ We realize this loop by starting with an
$n$-fold \textquotedblleft loop within a loop within a loop,\textquotedblright%
\ with the areas of the $n-2$ outermost annular regions set equal to
$\varepsilon$, and then letting $\varepsilon$ tend to zero. We then order the
$n-1$ crossings from outermost to innermost. We then make a Makeenko--Migdal
variation at the $k$th crossing, scaled by a factor of $k.$ (See Figure
\ref{5fold2.fig}.) The resulting change in the areas is%
\begin{align*}
a  &  \mapsto a-(n-1)t\\
b  &  \mapsto b+nt\\
c  &  \mapsto c-t,
\end{align*}
while the areas of the outer $n-2$ annular regions remain unchanged.
%TCIMACRO{\FRAME{ftbpFU}{2.981in}{2.4111in}{0pt}{\Qcb{The variation in the
%areas of $L_{5}(a,b,c)$}}{\Qlb{5fold2.fig}}{5fold2.eps}%
%width 2.981in;  height 2.4111in;  depth 0pt;  original-width 3.4039in;
%original-height 2.8591in;  cropleft "0";  croptop "1";  cropright "1";
%cropbottom "0";  filename '5fold2.eps';file-properties "XNPEU";}} }%
%BeginExpansion
\begin{figure}[ptb]%
\centering
\includegraphics[
height=2.4111in,
width=2.981in
]%
{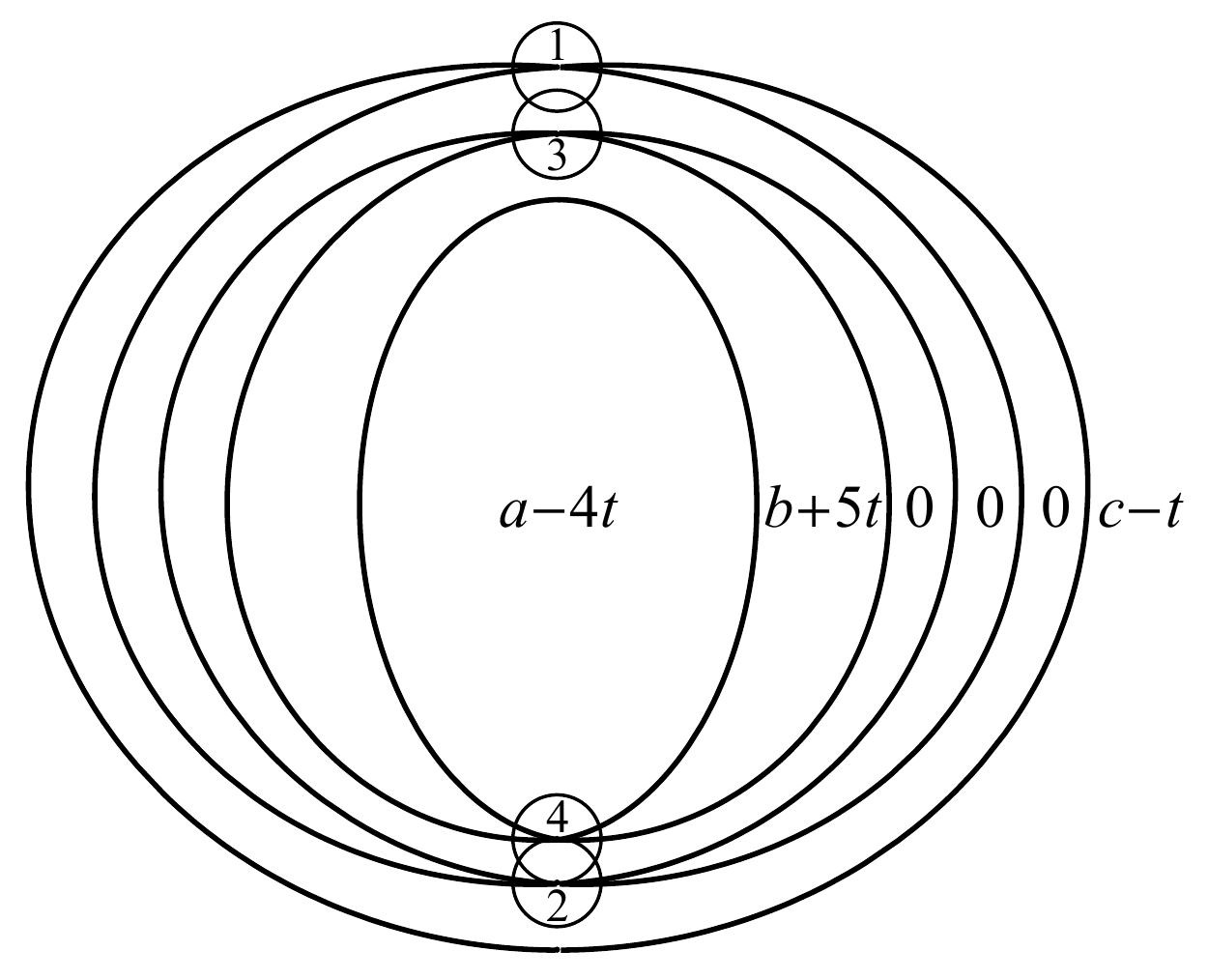}%
\caption{The variation in the areas of $L_{5}(a,b,c)$}%
\label{5fold2.fig}%
\end{figure}
%EndExpansion

If $c\geq a/n,$ the area of the outermost region will remain non-negative as
the area of the innermost region shrinks to zero. Thus, we may apply the
Makeenko--Migdal equation (\ref{variationWilson}) with the $n-2$ annular areas
equal to $\varepsilon$ and then let $\varepsilon$ tend to zero, giving a
finite-$N$ version of (\ref{WnPlus1}):%
\begin{align}
&  W_{n+1}^{N}(a,b,c)\nonumber\\
&  =W_{n}^{N}\left(  \frac{n+1}{n}a+b,c-\frac{a}{n}\right) \nonumber\\
&  -\sum_{k=1}^{n}k\int_{0}^{a/n}W_{k}^{N}(a+b+t,c-t)W_{n+1-k}^{N}%
(a-nt,b+(n+1)t,c-t)~dt\nonumber\\
&  +\mathrm{Cov}. \label{WnInduct}%
\end{align}
Here $W_{n}^{N}(a,b,c)=\mathbb{E}\left\{  \mathrm{tr}(\mathrm{hol}%
(L_{n}(a,b,c))\right\}  $ and $W_{n}^{N}(a,b)=\mathbb{E}\left\{
\mathrm{tr}(\mathrm{hol}(L_{n}(a,b)))\right\}  $ and \textquotedblleft%
$\mathrm{Cov}$\textquotedblright\ is a covariance term.

Then, as in the proof of the first two points of Theorem \ref{main.thm}, our
induction hypothesis, together with (\ref{covIneq}) and Proposition
\ref{varianceProd.prop}, allows us to take the $N\rightarrow\infty$ limit with
the covariance term going to zero. Thus, the large-$N$ limit of $W_{n+1}%
^{N}(a,b,c)$ exists for $c\geq a/n$ and satisfies (\ref{WnPlus1}). The
vanishing of the variance then follows by the same argument as in the proof of
Theorem \ref{main.thm}.

An entirely similar argument applies when $c<a/n,$ using a finite-$N$ version
of (\ref{WnPlus12}), establishing existence of the limit and vanishing of the
variance for $W_{n+1}(a,b,c)$ for all non-negative values of $a,$ $b,$ and
$c.$ The limit in (\ref{wnLim}) and vanishing of the variance then follow by
setting $b=0.$
\end{proof}

\section{The plane case, revisited\label{plane.sec}}

\subsection{The general result\label{planeGeneral.sec}}

In \cite{LevyMaster}, Thierry L\'{e}vy establishes the large-$N$ limit of
Wilson loop functionals for Yang--Mills theory in the plane, together with the
vanishing of the associated variances. (See Theorem 9.3.1 in \cite{LevyMaster}%
.)\ The limiting expectation values for loops with simple crossings are then
characterized in \cite{LevyMaster} by the large-$N$ Makeenko--Migdal equation
and another condition, labeled as Axiom $\Phi_{4}$ in Section 0 of
\cite{LevyMaster} and termed the \textit{unbounded face condition} in Theorem
2.3 of \cite{DHK2}.

The methods of the present paper give a new proof of some of these results,
making it possible to treat the cases of the plane and the sphere in a uniform
way. Specifically, the proof of Theorem \ref{main.thm} applies with minor
changes in the plane case. Furthermore, the analog of Theorem \ref{dn.thm} is
easily established by direct computation in the plane case, as we will
demonstrate shortly. We thus obtain a characterization of the master field on
the plane slightly different from the one in \cite{LevyMaster}; both
characterizations use the large-$N$ Makeenko--Migdal equation, but ours
replaces the unbounded face condition with the (simple!) formula for the
Wilson loop functional of a simple closed curve in the plane.

\begin{theorem}
\label{planeMain.thm}Theorem \ref{main.thm} holds also in the plane case.
\end{theorem}

\begin{proof}
We follow the same argument as in the case of the sphere, with a few minor
modifications. First, in Section \ref{selecting.sec}, we should choose $F_{0}$
to the unbounded face, so that the areas $a_{1},\ldots,a_{f-1}$ are finite
numbers. Theorem \ref{shrinkingLim.thm} still holds, with the same proof, in
the plane case, provided we interpret $\rho_{\left\vert F_{0}\right\vert }$ as
being the constant function $\mathbf{1}.$ (The normalization factor $Z$ is
then also equal to 1.) All other proofs go through without change---except
that in Section \ref{nfold.sec}, the condition $c\geq a/n$ always holds
because $c=\infty$.
\end{proof}

We now supply the proof of the plane case of Theorem \ref{dn.thm}. Of course,
the claim follows from results in \cite{LevyMaster} (see Theorem 6.3.1), but
we can give a elementary proof as follows.

\begin{proof}
[Proof of Theorem \ref{dn.thm} for $\mathbb{R}^{2}$]The function
$\mathrm{tr}(U)$ is an eigenfunction for the Laplacian on $U(N)$---with
respect to the scaled Hilbert--Schmidt metric in (\ref{innerProduct})---with
eigenvalue $-1.$ (See, e.g., Remark 3.4 in \cite{DHK1}.) It follows that for a
simple closed curve $C$ enclosing area $a,$ we have
\[
\mathbb{E}\left\{  \mathrm{tr}(\mathrm{hol}(C))\right\}  =\left.
e^{a\Delta/2}(\mathrm{tr}(U))\right\vert _{U=I}=e^{-a/2}.
\]
It is also possible to compute the (complex) variance of $\mathrm{tr}%
(\mathrm{hol}(C))$ explicitly. Using results from Section 3 of \cite{DHK1}, we
may easily compute that%
\[
\Delta\lbrack\mathrm{tr}(U)\overline{\mathrm{tr}(U)}]=-2\mathrm{tr}%
(U)\overline{\mathrm{tr}(U)}-\frac{2}{N^{2}}.
\]
(We have used that $\overline{\mathrm{tr}(U)}=\mathrm{tr}(U^{\ast
})=\mathrm{tr}(U^{-1}).$) Thus, the space of functions spanned by
$\mathrm{tr}(U)\overline{\mathrm{tr}(U)}$ and 1 is invariant under $\Delta,$
and the action of $\Delta$ on this space is represented by the matrix%
\[
\left(
\begin{array}
[c]{cc}%
-2 & 0\\
-2/N^{2} & 0
\end{array}
\right)  .
\]

We may exponentiate $a/2$ times this matrix using elementary formulas for the
exponential of $2\times2$ matrices (e.g., Exercise 5 or Exercises 6 and 7 in
\cite[Chapter 2]{HallBook}), with the result that%
\[
e^{a\Delta/2}\left[  \mathrm{tr}(U)\overline{\mathrm{tr}(U)}\right]
=e^{-a}\mathrm{tr}(U)\overline{\mathrm{tr}(U)}+\frac{e^{-a}-1}{N^{2}},
\]
so that
\begin{align*}
\mathbb{E}\left\{  \left\vert \mathrm{tr}(U)\right\vert ^{2}\right\}   &
=\left.  e^{a\Delta/2}\left[  \mathrm{tr}(U)\overline{\mathrm{tr}(U)}\right]
\right\vert _{U=I}\\
&  =e^{-a}+\frac{e^{-a}-1}{N^{2}}.
\end{align*}
Thus, the complex variance of $\mathrm{tr}(\mathrm{hol}(C))$ is%
\[
\mathbb{E}\left\{  \left\vert \mathrm{tr}(\mathrm{hol}(C))\right\vert
^{2}\right\}  -\left\vert \mathbb{E}\left\{  \mathrm{tr}(\mathrm{hol}%
(C))\right\}  \right\vert ^{2}=e^{-a}+\frac{e^{-a}-1}{N^{2}}-e^{-a},
\]
which goes to zero as $N$ tends to infinity.

More generally, Theorem 1.20 in \cite{DHK1} identifies the leading-order term
in the large-$N$ asymptotics of the $U(N)$ Laplacian, acting on
\textquotedblleft trace polynomials\textquotedblright\ (that is, sums of
products of traces of powers of $U$). This leading term satisfies a
\textit{first-order} product rule. Thus, the limiting heat operator is
multiplicative, from which it follows that the variance of any trace
polynomial vanishes as $N\rightarrow\infty.$ (See the last displayed formula
on p. 2620 of \cite{DHK1}.)
\end{proof}

\subsection{Example computations in the plane case\label{planeExamples.sec}}

We now illustrate the computation of large-$N$ Wilson loop functionals in the
plane case. We first compute the large-$N$ Wilson loop functionals
$W_{n}(a,b,c)$ from Section \ref{nfold.sec} in the plane case, which
corresponds to $c=\infty,$ for $n\leq3.$ The condition $c\geq a/n$ is always
satisfied in that case, and, as noted in the previous subsection, the Wilson
loop functional for a simple closed curve is $W_{1}(a,\infty)=e^{-a/2}.$

\textbf{The two-fold circle}. When $c=\infty,$ (\ref{WnPlus1}) takes the
following form:%
\begin{align*}
W_{2}(a,b,\infty)  &  =W_{1}(2a+b,\infty)-\int_{0}^{a}W_{1}(a+b+t,\infty
)W_{1}(a-t,\infty)~dt\\
&  =e^{-a-b/2}-\int_{0}^{a}e^{-(a+b+t)/2}e^{-(a-t)/2}~dt\\
&  =e^{-a-b/2}-e^{-a-b/2}\int_{0}^{a}~dt\\
&  =e^{-a-b/2}(1-a),
\end{align*}
which agrees with L\'{e}vy's formula for the \textquotedblleft loop within a
loop\textquotedblright\ in the appendix of\ \cite{LevyMaster}. When $b=0,$ we
get $W_{2}(a,\infty)=e^{-a}(1-a),$ which is the second moment of Biane's
distribution for the free unitary Brownian motion. (The moments may be found
in Lemma 3 of \cite{Bi1} or the Remark on p. 267 of \cite{Bi2}.)

\textbf{The three-fold circle}. When $n=3$ and $c=\infty,$ (\ref{w3induct})
becomes%
\begin{align*}
W_{3}(a,b,\infty)  &  =e^{-(b+3a/2)}(1-b-3a/2)\\
&  -\int_{0}^{a/2}e^{-(a+b+t)/2}e^{-(a-2t)-(b+3t)/2}(1-a+2t)~dt\\
&  -2\int_{0}^{a}e^{-(a+b+t)}(1-a-b-t)e^{-(a-2t)/2}~dt.
\end{align*}
All the $t$-terms in the exponents cancel, leaving us with $e^{-(b+3a/2)}$
times polynomials. Direct computation of the integral then gives%
\[
e^{-(b+3a/2)}\left(  1-3a+\frac{3}{2}a^{2}-b(1-a)\right)  ,
\]
which agrees with the $s=0$ case of the last of L\'{e}vy's examples of loops
with two crossings in the appendix of \cite{LevyMaster}. When $b=0,$ we get
\thinspace$e^{-3a/2}(1-3a+3a^{2}/2),$ which is the third moment of Biane's distribution.

\textbf{The trefoil.} We now analyze the large-$N$ Wilson loop functional for
the trefoil loop in the plane case, using the strategy outlined in Section
\ref{strategy.sec}. In this example, the loops $L_{1,j}$ and $L_{2,j}$
occurring on the right-hand side of the Makeenko--Migdal equation are simple
closed curves for all $j=1,2,3$, and the product of the Wilson loop
functionals simplifies as
\[
W(L_{1,j}(t))W(L_{2,j}(t))=e^{-a-(b+c+d)/2},
\]
where the dependence on $t$ and $j$ drops out. Thus, when we integrate from
$t=0$ to $t=1$, we obtain simply the constant value of the integrand. Now, the
signs in the Makeenko--Migdal variations in Figure \ref{trefoilshrink.fig} are
reversed from usual labeling. After adjusting for this and noting that the
coefficients of the variations in Figure \ref{trefoilshrink.fig} add to
$(b+c+d)/2$, we obtain
\[
W(L)=W_{2}(a+(b+c+d)/2,\infty)+\left(  \frac{b+c+d}{2}\right)
e^{-a-(b+c+d)/2}.
\]
Using the value for $W_{2}(a+(b+c+d)/2,\infty)$ computed above and simplifying
gives
\[
W(L)=e^{-a-(b+c+d)/2}(1-a),
\]
which agrees with the value of the master field for the trefoil in the
appendix of \cite{LevyMaster}.

\section{General compact surfaces\label{surfaces.sec}}

\subsection{Introduction\label{surfacesIntro.sec}}

We now consider the Yang--Mills measure on a compact surface $\Sigma,$
possibly with boundary. (See \cite{Sen97b} and also \cite{LevSurfaces}.) If
the boundary of $\Sigma$ is nonempty, we may optionally impose constraints on
the holonomy around the boundary components. The Makeenko--Migdal equation in
this setting was established rigorously in \cite{DGHK}.

Let us say that a loop in $\Sigma$ is \textbf{topologically trivial} if it is
contained in an (open) topological disk $U\subset\Sigma.$ All results in this
section pertain only to topologically trivial loops. For a topologically
trivial simple closed curve in $\Sigma,$ one may reasonably expect that the
analog of Theorem \ref{dn.thm} will hold, but the author is not aware of any
results in this direction.

Suppose, however, that we simply assume that the analog of Theorem
\ref{dn.thm} holds for topologically trivial simple closed curves in $\Sigma.$
We will then establish the analog of Theorem \ref{main.thm}---not for all
topologically trivial loops, but only for those satisfying a \textquotedblleft
smallness\textquotedblright\ condition. The reason for the smallness
assumption is as follows. Suppose $L$ is a loop in a topological disk $U$ and
let $F_{0}$ denote the face of $L$ that contains the complement of $U$ in
$\Sigma.$ Then in the first stage of deformation, following the procedure in
Section \ref{selecting.sec}, we perform a Makeenko--Migdal variation of $L$ to
a loop that winds $n$ times around a simple closed curve. This first variation can be done in such a way that \textit{the area of
}$F_{0}$\textit{ does not decrease}, so that the deformed loop can be chosen
to remain in $U$. In the second stage of deformation, we try to reduce---as in
Section \ref{nfold.sec}---from a curve that winds $n$ times around a simple
closed curve to a curve that winds only once around a simple closed curve. In
this second variation, the area of $F_{0}$ will unavoidably decrease and may
become zero, unless the original loop is \textquotedblleft
small.\textquotedblright\ If the area of $F_{0}$ becomes zero, the deformation
process becomes undefined, because the limit of the Wilson loop as $\left\vert
F_{0}\right\vert $ tends to zero is not easily evaluated (except when $F_{0}$
is a topological disk, i.e., when $\Sigma$ is a sphere).

On the other hand, suppose we assume that a natural strengthening of Theorem
\ref{dn.thm} holds for topologically trivial loops in $\Sigma$. Specifically,
suppose we simply \textit{assume} that Theorem \ref{dn.thm} holds not only for
topologically trivial simple closed curves in $\Sigma$, but also for curves
that wind $n$ times around a topologically trivial simple closed curve. Then
the second stage of deformation in the previous paragraph is not needed. Under
this assumption, therefore, we will be able to prove Theorem \ref{main.thm}
for \textit{all} topologically trivial loops in $\Sigma.$

Let us consider what the previous discussion means for the study of the master
field on a general compact surface $\Sigma.$ The analysis will have two
stages, as in the $S^{2}$ case: A direct calculation for the case of simple
closed curves, and a reduction of general curves to the simple closed case
using the Makeenko--Migdal equation. If we hope to use the just-discussed
results for general surface, \textit{more must be put into the direct
calculation stage}.

For a simple closed curve $C$ in $\Sigma,$ Sengupta's formula gives a
probability measure $\mu_{C}$ on $U(N),$ which describes the distribution of
the holonomy of a random connection around $C.$ If we ultimately wish to
analyze \textit{all} topologically trivial loops with simple crossing, then in
the direct calculation stage of analysis, we will have to look not just at the
first moment of $\mu_{C},$
\begin{equation}
\int_{U(N)}\mathrm{tr}(U)~d\mu_{C}(U), \label{firstMoment}%
\end{equation}
but also at the higher moments,%
\begin{equation}
\int_{U(N)}\mathrm{tr}(U^{n})~d\mu_{C}(U), \label{higherMoment}%
\end{equation}
which are just the Wilson loop functionals for curves that wind $n$ times
around $C.$ If one could establish directly (i.e., without using the
Makeenko--Migdal equation) the existence of the large-$N$ limit and the
vanishing of the variance for the quantities in (\ref{firstMoment}) and
(\ref{higherMoment}), then the results of this section would apply to give the
existence of the limit and vanishing of the variance for all topologically
trivial loops with simple crossings.

\subsection{Statements}

We begin by stating the analog of Theorem \ref{dn.thm} for topologically
trivial loops in $\Sigma$ as a conjecture.

\begin{conjecture}
\label{surfaces.conjecture}Consider a surface $\Sigma$ of area $\mathrm{area}%
(\Sigma).$ Let $U\subset\Sigma$ be a topological disk and consider a simple
closed curve $C$ in $U$. In Sengupta's formula for the Yang--Mills measure on
$\Sigma$, let us assign area $a$ to the interior of $C$ and area
\[
c:=\mathrm{area}(\Sigma)-a
\]
to the exterior of $C,$ for any $a<\mathrm{area}(\Sigma).$ Then the limit%
\[
\lim_{N\rightarrow\infty}\mathbb{E}\left\{  \mathrm{tr}(\mathrm{hol}%
(C))\right\}
\]
exists and depends continuously on $a$ and $c,$ and
\[
\lim_{N\rightarrow\infty}\mathrm{Var}\left\{  \mathrm{tr}(\mathrm{hol}%
(C))\right\}  =0.
\]

\end{conjecture}

We now introduce a notion of \textquotedblleft smallness\textquotedblright%
\ for a general loop $L$ in $U$ with simple crossings. We may cut a loop $L$
in $U$ at a crossing, obtaining two loops $L_{1}$ and $L_{2}$. We may then cut
either $L_{1}$ or $L_{2}$ at one of its crossings, and so on. We refer to any
loop that can be obtained by a finite sequence of such cuts as a
\textbf{subloop} of $L.$ In particular, $L$ is a subloop of itself
corresponding to making zero cuts. We label the faces of $L$ as $F_{0}%
,\ldots,F_{f-1}$ where $F_{0}$ is the face containing the complement $U^{c}$
of $U.$

For a loop $L$ in $U$ and a point $x$ that is in $U$ but not in $L,$ we define
the winding number of $L$ around $x$ by considering the homotopy class of $L$
in $U\setminus\{x\},$ with respect to a fixed orientation on $U.$ Next, we
define, for each face $F_{j}$ of $L,$%
\begin{equation}
\left\vert w\right\vert _{\mathrm{\max}}(F_{j})=\max_{L^{\prime},j}\left\vert
w(L^{\prime},F_{j})\right\vert , \label{wMaxDef}%
\end{equation}
where the maximum ranges over all subloops $L^{\prime}$ of $L$ and all
$j=1,\ldots,f-1,$ and where $w(L^{\prime},F)$ is the winding number of
$L^{\prime}$ around $F.$ Finally, if $a_{j}$ is the area of $F_{j},$ we define%
\begin{equation}
A=a_{1}+a_{2}+\cdots+a_{f-1}=\mathrm{area}(\Sigma)-\mathrm{area}(F_{0}).
\label{Adef}%
\end{equation}

\begin{remark}
One may bound $\left\vert w\right\vert _{\max}$ as follows. Suppose for some
$k,$ one can travel from any face of $L$ to $F_{0}$ while crossing $L$ at most
$k$ times. Then the same is true of any subloop $L^{\prime}$ of $L.$ Since the
winding number of $F_{0}$ is zero and $w(L^{\prime},F)$ changes by one each
time we cross $L^{\prime},$ we conclude that $\left\vert w\right\vert _{\max}$
is at most $k.$
\end{remark}

\begin{theorem}
\label{surfacesMain.thm}Assume Conjecture \ref{surfaces.conjecture}. Let $L$
be a loop traced out on a graph in a topological disk $U\subset\Sigma$ with
simple crossings. Then Theorem \ref{main.thm} holds for $L$, provided $L$
satisfies the \textquotedblleft smallness\textquotedblright\ assumption
\begin{equation}
A\left\vert w\right\vert _{\mathrm{\max}}<\mathrm{area}(\Sigma),
\label{smallness}%
\end{equation}
where $\left\vert w\right\vert _{\max}$ and $A$ are defined in (\ref{wMaxDef})
and (\ref{Adef}), respectively.
\end{theorem}

Note that since $L$ is contained in a disk, $L$ is certainly homotopically
trivial in $\Sigma.$ Thus, Theorem \ref{surfacesMain.thm} does not tell us
anything about homotopically nontrivial loops. Since $U$ is a topological
disk, Theorem \ref{MMspan.thm} applies in this context, provided we compute
the winding numbers of the faces of $L$ by regarding $L$ as a loop in $U.$

We next state a natural extension of Conjecture \ref{surfaces.conjecture}.

\begin{conjecture}
\label{surfaces2.conjecture}Continuing with the notation of Conjecture
\ref{surfaces.conjecture}, let $C^{(n)}$ denote the loop obtained by traveling
$n$ times around $C.$ Then for all $n\in\mathbb{Z},$ the limit%
\[
\lim_{N\rightarrow\infty}\mathbb{E}\left\{  \mathrm{tr}(\mathrm{hol}%
(C^{(n)}))\right\}
\]
exists and depends continuously on $a$ and $c,$ and
\[
\lim_{N\rightarrow\infty}\mathrm{Var}\left\{  \mathrm{tr}(\mathrm{hol}%
(C^{(n)}))\right\}  =0.
\]

\end{conjecture}

Assuming Conjecture \ref{surfaces2.conjecture} holds, we can prove Theorem
\ref{main.thm} for topologically trivial loops in $\Sigma,$ \textit{without}
imposing a smallness assumption.

\begin{theorem}
\label{surfacesMain2.thm}Assume Conjecture \ref{surfaces2.conjecture}. Let $L$
be a loop with only simple crossing, traced out on a graph in a topological
disk $U$ in $\Sigma$. Then Theorem \ref{main.thm} holds for $L,$ without any
smallness assumption on $L.$
\end{theorem}

\subsection{Proofs}

Let $U$ be a topological disk in $\Sigma$ and let $L$ be a loop traced out on
a graph in $U$ and having only simple crossings. We follow the deformation
process in Section \ref{selecting.sec}, with the face $F_{0}$ whose area does
not decrease taken to be the face of $L$ containing $\Sigma\setminus U.$ The
obvious analog of Theorem \ref{shrinkingLim.thm} continues to hold (see Lemma
\ref{surfaceLimit.lem} below), so that at the end of the deformation process,
we obtain a loop that winds $n$ times around a simple closed curve in $U.$ We
may apply the same deformation process to all the subloops generated by the
Makeenko--Migdal equation, and again the area of $F_{0}$ will not decrease.
Thus, if we assume the stronger assertion in Conjecture
\ref{surfaces2.conjecture}, we may prove Theorem \ref{surfacesMain2.thm}
precisely as in Section \ref{induction.sec}, but using Conjecture
\ref{surfaces2.conjecture} in place of Theorem \ref{nfold.thm}.

If, on the other hand, we wish to assume only the weaker assertion in
Conjecture \ref{surfaces.conjecture}, we must proceed to the second stage of
analysis, deforming the loop that winds $n$ times around a simple closed curve
to one that winds once around a simple closed curve, as in Section
\ref{nfold.sec}. In this stage of analysis, the area of $F_{0}$ will
unavoidably decrease, unless $\left\vert n\right\vert =1$. If the area of
$F_{0}$ becomes zero during the deformation, the whole process fails, because
there is no easy way to compute the limit of a Wilson loop functional as the
area of $F_{0}$ goes to zero. (See (\ref{SenguptaInSurface}), in which the
limit as $\left\vert F_{0}\right\vert $ goes to zero is not easily computed.
The only case where this limit is easy to evaluate is when $F_{0}$ is a
topological disk, that is, when $\Sigma$ is a sphere.)

In the case of a curve winding $n$ times around a simple closed curve, there
is a simple condition (Theorem \ref{surfacesNfold.thm}) on the $n$ and the
enclosed area guaranteeing that the area of $F_{0}$ will remain positive. We
must then show that the smallness condition in Theorem \ref{surfacesMain.thm}
guarantees that each \textquotedblleft$n$-fold circle\textquotedblright%
\ generated by $L$ will satisfy the hypotheses of Theorem
\ref{surfacesNfold.thm}.

\begin{theorem}
\label{surfacesNfold.thm}Let the notation be as in Conjecture
\ref{surfaces.conjecture}. For each nonzero integer $n,$ let $L_{n}(a,c)$
denote the loop that winds $n$ times around $C.$ Assuming Conjecture
\ref{surfaces.conjecture}, Theorem \ref{main.thm} holds for $L_{n}(a,c),$
provided that
\begin{equation}
\left\vert n\right\vert a<\mathrm{area}(\Sigma). \label{naArea}%
\end{equation}

\end{theorem}

\begin{proof}
It is harmless to assume $n>0.$ Following the proof of Theorem \ref{nfold.thm}
in the sphere case, we deform $L_{n}(a,c)$ into a loop of the form
$L_{n}(a,b,c).$ Let $\mathbf{a}=(a,b,c)$ denote the vector of areas of the
faces of $L_{n}(a,b,c)$ and let $\mathbf{w}=(w_{1},w_{2},w_{3})$ be the
associated vector of winding numbers, viewing $L_{n}(a,b,c)$ as a loop in the
disk $U,$ with $w_{3}=0.$ We will actually prove Theorem
\ref{surfacesNfold.thm} for the loops $L_{n}(a,b,c)$ by induction on $n>0,$
under the assumption that
\begin{equation}
\mathbf{a}\cdot\mathbf{w}<\mathrm{area}(\Sigma), \label{nabArea}%
\end{equation}
which reduces to (\ref{naArea}) when $n>0$ and $b=0.$ When $n=1$, the loop
$L_{1}(a,b,c)$ is (by definition) the same as $L_{1}(a,b+c),$ so that the
desired result is just Conjecture \ref{surfaces.conjecture}. For $n\geq2,$ we
may follow the same sort of inductive argument as in the sphere case, provided
that we never shrink the area of the \textquotedblleft$c$\textquotedblright%
\ region to zero.

Take $n\geq2$ and assume that Theorem \ref{surfacesNfold.thm} holds for loops
of the form $L_{k}(a,b,c)$ satisfying (\ref{nabArea}), with $k<n.$ Consider a
loop $L_{n}(a,b,c)$ satisfying (\ref{nabArea}) and deform it into the loop
$L_{n}(a(t),b(t),c(t))$ with $0\leq t\leq a/(n-1),$ as in Section
\ref{nfold.sec}. As we vary the values of $a,$ $b,$ and $c,$ the values of
$\mathrm{area}(\Sigma)$ and $\mathbf{a}\cdot\mathbf{w}$ remains constant---as
can be seen explicitly or as a consequence of Theorem \ref{MMspan.thm}. Thus,
by (\ref{nabArea}), we have%
\begin{equation}
\mathbf{a}\cdot\mathbf{w}=na(t)+(n-1)b(t)<a(t)+b(t)+c(t). \label{awn}%
\end{equation}
Since $n\geq2,$ (\ref{awn}) tells us that $c(t)>(n-1)a(t)+(n-2)b(t)>0.$

Thus, $c(t)$ remains positive as $t$ approaches $a/(n-1),$ when we obtain a
loop of the form $L_{n-1}(a^{\prime},c^{\prime}),$ with
\[
a^{\prime}=\left.  b\left(  t\right)  \right\vert _{t=a/(n-1)}=b+\frac{n}%
{n-1}a.
\]
We can then see explicitly that the value of $\mathbf{a}\cdot\mathbf{w}$ for
$L_{n-1}(a^{\prime},c^{\prime})$ is the same as for $L_{n}(a,b,c),$ namely
$na+(n-1)b.$ Thus, by induction, Theorem \ref{surfacesMain.thm} holds for
$L_{n-1}(a^{\prime},c^{\prime}).$ Furthermore, the loops $L_{1,j}(t)$ and
$L_{2,j}(t)$ obtained from the Makeenko--Migdal equation will be
$L_{k}(a(t),b(t),c(t))$ or $L_{n-k}(a(t)+b(t),c(t)),$ with $1\leq k<n.$ It is
then easy to see that the value of $\mathbf{a}\cdot\mathbf{w}$ for these loops
is no bigger than for $L_{n}(a(t),b(t),c(t)),$ which is the same as for
$L_{n}(a,b,c).$ Thus, $L_{1,j}(t)$ and $L_{2,j}(t)$ satisfy (\ref{nabArea})
and by induction, Theorem \ref{surfacesMain.thm} holds for these loops as well.

From this point, the argument is the same as in the sphere case. In
particular, since we have ensured that $c(t)$ remains positive as we deform
the areas of $L_{n}(a,b,c)$, we may apply (\ref{WnInduct}) to compute the
Wilson loop functional for $L_{n}(a,b,c).$
\end{proof}

We now prove Theorem \ref{surfacesMain.thm}. Following the logic in Section
\ref{general.sec}, we first deform $L$ into a loop of the form $L_{n}(a,c).$
Since the area of $F_{0}$ increases during this process, there is no
obstruction to carrying out this first step in the analysis of $L.$
Nevertheless, we require a smallness assumption on $L$ that will ensure that
the limiting loop $L_{n}(a,c)$ will satisfy the hypothesis of Theorem
\ref{surfacesNfold.thm}. This smallness assumption must also be inherited by
the loops $L_{1,j}(t)$ and $L_{2,j}(t)$ occurring on the right-hand side of
the Makeenko--Migdal equation, so that these loops can be analyzed by
induction on the number of crossings.

\begin{proof}
[Proof of Theorem \ref{surfacesMain.thm}]We proceed by induction on the number
$k$ of crossings. If $k=0,$ the result is Conjecture \ref{surfaces.conjecture}%
. Assume, then, that the result holds for loops with fewer than $k$ crossings
and consider a loop $L$ with $k$ crossings. As in Lemma \ref{surfaceLimit.lem}%
, we deform $L$ into a loop $L_{n}(a,c).$ By the Makeenko--Migdal equation,
the variation of the Wilson loop functional will involve loops of the form
$L_{1,j}(t)$ and $L_{2,j}(t),$ all of which have fewer than $k$ crossings. We
need to verify (1) that $L_{n}(a,c)$ satisfies $\left\vert n\right\vert
a<\mathrm{area}(\Sigma)$ and (2) that $L_{1,j}(t)$ and $L_{2,j}(t)$ both
satisfy the smallness assumption (\ref{smallness}). If so, we may apply
Theorem \ref{surfacesNfold.thm} to the loops $L_{n}(a,c)$ and our induction
hypothesis to $L_{1,j}(t)$ and $L_{2,j}(t)$ and the argument is then the same
as in the sphere case.

For Point (1), we note that the value of $n$ is (Lemma \ref{surfaceLimit.lem})
the winding number of $L$ around $F_{1},$ while the value of $a$ is the
limiting value of $a_{1}(t)$ (the area of $F_{1}$) as $t$ approaches 1. Now,
on the one hand, the value of $\mathbf{a}\cdot\mathbf{w}$ for $L(t)$ is
independent of $t,$ by Theorem \ref{MMspan.thm}. On the other hand,
$\mathbf{a}\cdot\mathbf{w}$ approaches the value $na$ as $t$ approaches 1,
since $a_{1}(t)\rightarrow a$ and $a_{j}(t)\rightarrow0$ for $j\geq2.$ Thus,
$na=\mathbf{a}\cdot\mathbf{w}$. But from the definitions (\ref{Adef}) and
(\ref{wMaxDef}) of $A$ and $\left\vert w\right\vert _{\mathrm{\max}},$ we have
$\left\vert \mathbf{a}\cdot\mathbf{w}\right\vert \leq A\left\vert w\right\vert
_{\max}$ and thus
\[
\left\vert n\right\vert a=\left\vert \mathbf{a}\cdot\mathbf{w}\right\vert \leq
A\left\vert w\right\vert _{\max}<\mathrm{area}(\Sigma).
\]

For Point (2), we note that for the loops $L(t),$ the area of $F_{0}$ is
always increasing (Proposition \ref{shrinkAllButTwo.prop}), meaning that
$A(t)$ is decreasing. Thus, for all $i$ and $j,$ we have
\[
A_{L_{i,j}(t)}\leq A_{L(t)}\leq A.
\]
Furthermore, since every subloop of $L_{i,j}$ is also a subloop of $L,$ we see
that $|w_{L_{i,j}}|_{\max}\leq|w_{L}|_{\max}.$ Thus, we have%
\[
A_{L_{i,j}(t)}|w_{L_{i,j}}|_{\max}\leq A|w_{L}|_{\max}<\mathrm{area}(\Sigma).
\]
Having verified these two points, the proof now proceeds as in the sphere case.
\end{proof}

It remains only to verify that the analog of Theorem \ref{shrinkingLim.thm}
holds for loops in $U\subset\Sigma.$

\begin{lemma}
\label{surfaceLimit.lem}Consider the Yang--Mills measure on $\Sigma$ for an
arbitrary connected compact Lie group $K.$ Let $L$ be a loop traced out on a
graph in $U\subset\Sigma$ and having only simple crossings. Denote the number
of faces of $L$ by $f$ and label the faces as $F_{0},F_{1},F_{2}%
,\ldots,F_{f-1},$ where $F_{0}$ is the face containing $U^{c}.$ Suppose we
vary the areas of the faces as a function of a parameter $t\in\lbrack0,1)$ in
such a way that as $t\rightarrow1,$ the areas of $F_{2},\ldots,F_{f-1}$ tend
to zero, while the areas of $F_{0}$ and $F_{1}$ approach non-negative real
numbers $c$ and $a,$ respectively, with $a>0.$ Then%
\[
\lim_{t\rightarrow1}\mathbb{E}\{\mathrm{tr}(\mathrm{hol}(L))\}=\mathbb{E}%
\left\{  \mathrm{tr}(\mathrm{hol}(L_{n}(a,c)))\right\}  ,
\]
where $n$ is the winding number of $L$ around $F_{1}$.
\end{lemma}

\begin{proof}
Let $\mathbb{G}$ be a minimal graph (necessarily connected) in which $L$ can
be traced. Before we can apply Sengupta's formula, we must embed $\mathbb{G}$
into an admissible graph $\mathbb{G}^{\prime}$, that is, one that contains the
boundary of $\Sigma$ and each of whose faces is a topological disk. Actually,
by the Jordan curve theorem, all the faces of $\mathbb{G}$ other than $F_{0}$
will automatically be disks. It is then possible to construct $\mathbb{G}%
^{\prime}$ by adding new edges entirely in the closure of $F_{0}.$ (See
Section 1.2 of \cite{LevSurfaces}.) Thus, the faces of $\mathbb{G}^{\prime}$
may be chosen to be of the form $F_{0}^{\prime},F_{1},\ldots,F_{f-1},$ where
$F_{0}^{\prime}$ is a subset of $F_{0}$ having the same area as $F_{0}.$ Let
us divide the edge variables for $\mathbb{G}^{\prime}$ into the edge variables
$\mathbf{x}$ corresponding to the original graph $\mathbb{G}$ and the
remaining edge variables $\mathbf{y}.$ Then integration with respect to the
Yang--Mills measure for the graph $\mathbb{G}^{\prime}$ in $\Sigma$ may be
written as%
\begin{equation}
\int_{K^{e^{\prime}}}f(\mathbf{x},\mathbf{y})~d\mu_{\Sigma}^{\mathbb{G}%
^{\prime}}=\frac{1}{Z}\int_{K^{e^{\prime}-e}}\rho_{\left\vert F_{0}\right\vert
}(\mathrm{hol}_{F_{0}^{\prime}}(\mathbf{x,y}))\int_{K^{e}}f(\mathbf{x}%
,\mathbf{y})~d\mu_{\mathrm{plane}}^{\mathbb{G}}(\mathbf{x})~d\mathbf{y},
\label{SenguptaInSurface}%
\end{equation}
where $e$ and $e^{\prime}$ are the number of edges of $\mathbb{G}$ and
$\mathbb{G}^{\prime}$, respectively. Furthermore, we may write%
\[
\mathrm{hol}_{F_{0}^{\prime}}(\mathbf{x,y})=\mathrm{hol}_{F_{0}}%
(\mathbf{x})g(\mathbf{y})
\]
for some word $g(\mathbf{y})$ in the $\mathbf{y}$ variables. Here, for
notational simplicity, we consider the unconditional Yang--Mills measure, but
a similar argument applies if there are constraints on the holonomies around
the boundary components of $\Sigma.$

In the case that $f$ is the trace of the holonomy of $L,$ we may imitate the
proof of Theorem \ref{shrinkingLim.thm} to obtain
\begin{align*}
&  \mathbb{E}\left\{  \mathrm{tr}(\mathrm{hol}(L))\right\}  =\frac{1}{Z}%
\int_{K^{e^{\prime}-e}}\int_{K^{f-1}}\mathrm{tr}(w_{1}(h_{1},\ldots
,h_{f-1}))\\
&  \times\rho_{\left\vert F_{0}\right\vert }(w_{0}(h_{1},\ldots,h_{f-1}%
)g(\mathbf{y}))\left(  \prod_{i=1}^{f-1}\rho_{\left\vert F_{i}\right\vert
}(h_{i})\right)  dh_{1}~\ldots dh_{f-1}~d\mathbf{y}.
\end{align*}
As in that proof, if we let $t\rightarrow1,$ we obtain%
\begin{equation}
\lim_{t\rightarrow1}\mathbb{E}\left\{  \mathrm{tr}(\mathrm{hol}(L))\right\}
=\frac{1}{Z}\int_{K^{e^{\prime}-e}}\int_{K^{f-1}}\mathrm{tr}(h_{1}^{n}%
)\rho_{c}(h_{1}g(\mathbf{y}))\rho_{a}(h_{1})~dh_{1}~d\mathbf{y},
\label{surfLimNfold}%
\end{equation}
where $n$ is the winding number of $L$ around $F_{1}.$ But the right-hand side
of (\ref{surfLimNfold}) is just Sengupta's formula for the Wilson loop
functional for the loop that winds $n$ times around the boundary of $F_{0}$
(i.e., the outer boundary of $\mathbb{G}$), enclosing areas $a$ and $c.$
\end{proof}

\section{Acknowledgments}

The author thanks Bruce Driver, David Galvin, Todd Kemp, Thierry L\'{e}vy,
Karl Liechty, Andy Putman, Ambar Sengupta, and Steve Zelditch for valuable
discussions. The author is also grateful to Franck Gabriel for a detailed
reading of the entire manuscript, along with many corrections and helpful
suggestions. Finally, the author thanks the two referees, whose careful
reading and thoughtful comments have improved the manuscript immensely.

\end{document}